\newif\ifHandout%

\documentclass{article}

\usepackage{
amsmath,%
amsthm,%
amsfonts,%
longtable,%
verbatim,%
xspace,%
multicol,%
tikz,%
framed,%
}

\usepackage[margin=1in]{geometry}

\usepackage{pifont}

\usepackage[colorlinks,citecolor=blue,urlcolor=blue]{hyperref}

\usepackage{tikz}
\usetikzlibrary{shapes,arrows,positioning}

\usepackage[T1]{fontenc}
\usepackage{inputenc}
\usepackage{geometry}
\usepackage{float}
\usepackage{ifthen}
\usepackage{amsbsy}
\usepackage{amssymb}
\usepackage{amsmath}
\usepackage{amsthm}
\usepackage{graphicx}
\usepackage{setspace}
\usepackage{esint}
\usepackage{comment}
\usepackage{mathtools}
\usepackage{xcolor}
\DeclareMathOperator*{\argmin}{arg\,min} %
\DeclareMathOperator*{\argmax}{arg\,max}

\def\ddefloop#1{\ifx\ddefloop#1\else\ddef{#1}\expandafter\ddefloop\fi}
\def\ddef#1{\expandafter\def\csname bb#1\endcsname{\ensuremath{\mathbb{#1}}}}
\ddefloop ABCDEFGHIJKLMNOPQRSTUVWXYZ\ddefloop
\def\ddefloop#1{\ifx\ddefloop#1\else\ddef{#1}\expandafter\ddefloop\fi}
\def\ddef#1{\expandafter\def\csname b#1\endcsname{\ensuremath{\mathbf{#1}}}}
\ddefloop ABCDEFGHIJKLMNOPQRSTUVWXYZ\ddefloop
\def\ddef#1{\expandafter\def\csname c#1\endcsname{\ensuremath{\mathcal{#1}}}}
\ddefloop ABCDEFGHIJKLMNOPQRSTUVWXYZ\ddefloop
\def\ddef#1{\expandafter\def\csname h#1\endcsname{\ensuremath{\widehat{#1}}}}
\ddefloop ABCDEFGHIJKLMNOPQRSTUVWXYZ\ddefloop
\def\ddef#1{\expandafter\def\csname hc#1\endcsname{\ensuremath{\widehat{\mathcal{#1}}}}}
\ddefloop ABCDEFGHIJKLMNOPQRSTUVWXYZ\ddefloop
\def\ddef#1{\expandafter\def\csname t#1\endcsname{\ensuremath{\widetilde{#1}}}}
\ddefloop ABCDEFGHIJKLMNOPQRSTUVWXYZ\ddefloop
\def\ddef#1{\expandafter\def\csname tc#1\endcsname{\ensuremath{\widetilde{\mathcal{#1}}}}}
\ddefloop ABCDEFGHIJKLMNOPQRSTUVWXYZ\ddefloop

\usepackage{prettyref}
\newcommand{\pref}[1]{\prettyref{#1}}

\newcommand{\Tcal}{\mathcal{T}}
\newcommand{\Rcal}{\mathcal{R}}
\newcommand{\Hcal}{\mathcal{H}}
\newcommand{\Fcal}{\mathcal{F}}
\usepackage{nicefrac}
\usepackage{subcaption}
\usepackage{pgfplots}

\newcommand{\shull}{\ensuremath{\text{star}}}

\makeatletter
\newtheorem*{rep@theorem}{\rep@title}
\newcommand{\newreptheorem}[2]{%
\newenvironment{rep#1}[1]{%
 \def\rep@title{#2 \ref{##1}}%
 \begin{rep@theorem}}%
 {\end{rep@theorem}}}
\makeatother

\usepackage{url}

\usepackage{cleveref}
\newtheorem{theorem}{Theorem}

\newreptheorem{theorem}{Theorem}
\newtheorem{corollary}[theorem]{Corollary}
\newreptheorem{corollary}{Corollary}
\newtheorem{lemma}[theorem]{Lemma}

\newtheorem{definition}{Definition}
\newtheorem{example}{Example}
\newtheorem{remark}{Remark}

\crefname{equation}{}{}
\crefname{proposition}{Proposition}{Propositions}
\crefname{appendix}{Appendix}{Appendices}
\usepackage{autonum}

\newcommand{\kibitz}[2]{\ifnum\Comments=1{\color{#1}{#2}}\fi}

\newcommand{\R}{\mathbb{R}}

\newcommand{\eg}{e.g., \xspace}
\newcommand{\ie}{i.e.,\xspace}

\newcommand{\mcX}{{\mathcal X}}

\newcommand{\mcE}{{\mathcal E}}

\newcommand{\ba}{\begin{array}}
\newcommand{\ea}{\end{array}}
\newcommand{\bs}{\begin{align}\begin{split}\nonumber}
\newcommand{\bsnumber}{\begin{align}\begin{split}}
\newcommand{\es}{\end{split}\end{align}}

\renewcommand{\[}{\left[}

\newtheorem{assumption}{ASSUMPTION}

\DeclareMathOperator{\E}{\mathbb{E}}

\def\balign#1\ealign{\begin{align}#1\end{align}}
\def\balignat#1\ealign{\begin{alignat}#1\end{alignat}}
\def\bitemize#1\eitemize{\begin{itemize}#1\end{itemize}}
\def\benumerate#1\eenumerate{\begin{enumerate}#1\end{enumerate}}
\newenvironment{talign}
 {\csname align\endcsname}
 {\endalign}
\def\balignt#1\ealignt{\begin{talign}#1\end{talign}}%

\newcommand{\Bcal}{\ensuremath{{\cal B}}}
\newcommand{\Qcal}{\ensuremath{{\cal Q}}}
\newcommand{\Gcal}{\ensuremath{{\cal G}}}
\newcommand{\clos}{\ensuremath{\mathrm{cl}}}
\usepackage{natbib}

\newcommand{\prns}[1]{\left(#1\right)}
\newcommand{\braces}[1]{\left\{#1\right\}}
\newcommand{\bracks}[1]{\left[#1\right]}
\DeclareMathSymbol{\nmid}{\mathrel}{AMSb}{"2D}

\renewcommand{\tilde}{\widetilde}

\DeclareMathSymbol{\varnothing}{\mathord}{AMSb}{"3F}

\ifHandout
  
  \newenvironment{solution}
    {\smallskip\par\noindent\emph{Solution: }\color{white}}
    {}
\else

\fi

\colorlet{shadecolor}{gray!35}

\newtheorem*{scholium}{Scholium}

\newenvironment{comments}{}{}

\newcommand{\myrule}{\nobreak\par\noindent\mbox{}\hspace*{.333333\textwidth}%
            \rule{.333333\textwidth}{.01in}\hspace*{.33333\textwidth}\par}
\makeatletter
  \g@addto@macro\enddefinition\myrule
  \g@addto@macro\endexample\myrule
  \g@addto@macro\endcomments\myrule
\makeatother

\makeatletter
  \define@key{hide}{scholium}[true]{}
  \define@key{hide}{proof}[true]{\renewenvironment{proof}{\comment}{\endcomment}}
  \define@key{hide}{lemma}[true]{\renewenvironment{lemma}{\comment}{\endcomment}}
  \define@key{hide}{comments}[true]{\renewenvironment{comments}{\comment}{\endcomment}}
  \define@key{hide}{cor}[true]{}
  \define@key{hide}{definition}[true]{\renewenvironment{definition}{\comment}{\endcomment}}
  \define@key{hide}{example}[true]{\renewenvironment{example}{\comment}{\endcomment}}
  \define@key{hide}{theorem}[true]{\renewenvironment{theorem}{\comment}{\endcomment}}

\newcommand{\varemdash}[1][10pt]{%
  \makebox[#1]{\leaders\hbox{---}\hfill\kern0pt}%
}

  \define@key{show}{scholium}[true]{\define@key{hide}{scholium}{}}
  \define@key{show}{proof}[true]{\define@key{hide}{proof}{}}
  \define@key{show}{lemma}[true]{\define@key{hide}{lemma}{}}
  \define@key{show}{comments}[true]{\define@key{hide}{comments}{}}
  \define@key{show}{cor}[true]{\define@key{hide}{cor}{}}
  \define@key{show}{definition}[true]{\define@key{hide}{definition}{}}
  \define@key{show}{example}[true]{\define@key{hide}{example}{}}
  \define@key{show}{theorem}[true]{\define@key{hide}{theorem}{}}

\makeatother

\begin{document}

\title{Source Condition Double Robust Inference on Functionals of Inverse Problems}

\author{Andrew Bennett\\
Cornell University\\
\texttt{awb222@cornell.edu}
\and
Nathan Kallus\\
Cornell University\\
\texttt{kallus@cornell.edu}
\and
Xiaojie Mao\\
Tsinghua University\\
\texttt{maoxj@sem.tsinghua.edu.cn}
\and
Whitney Newey\\
MIT\\ 
\texttt{wnewey@mit.edu}
\and
Vasilis Syrgkanis\\
Stanford University\\
\texttt{vsyrgk@stanford.edu}
\and
Masatoshi Uehara\\
Cornell University\\
\texttt{mu223@cornell.edu}}

\maketitle
\begin{abstract}
We consider estimation of parameters defined as linear functionals of solutions to linear inverse problems. Any such parameter admits a doubly robust representation that depends on the solution to a dual linear inverse problem, where the dual solution can be thought as a generalization of the inverse propensity function. We provide the first source condition double robust inference method that ensures asymptotic normality around the parameter of interest as long as either the primal or the dual inverse problem is sufficiently well-posed, without knowledge of which inverse problem is the more well-posed one. Our result is enabled by novel guarantees for iterated Tikhonov regularized adversarial estimators for linear inverse problems, over general hypothesis spaces, which are developments of independent interest.
\end{abstract}

\section{Introduction}\label{sec: intro}

Many important problems in social and biomedical sciences can be formulated as the estimation of linear functionals of unknown functions that are defined as solutions to linear inverse problems. Examples include nonparametric instrumental variable (IV) regression problems \citep[e.g., ][]{newey2013nonparametric,ai2012semiparametric,ai2003efficient,chen2015sieve}, missing-not-at-random problems \citep[e.g., ][]{d2010new,miao2015identification,LiMiao2022}, causal inference with unmeasured confounders in the presence of proxy variables (a.k.a. proximal causal inference) \citep[e.g., ][]{miao2018a,cui2020semiparametric,deaner2018proxy,kallus2021causal}, partially linear regression problems with endogenous regressors \citep[e.g., ][]{chen2021robust,bennett2022inference,ai2007estimation}, and off-policy evaluation in confounded contextual bandits or partially observable Markov decision processes \citep[e.g., ][]{tennenholtz2020off,shen2022optimal,bennett2021proximal,Miao2022OffPolicy}. 

All these problems are encompassed by the following general statistical estimation problem: we are given data that contain samples of the random variable $W$ and our parameter of interest is defined as:
\begin{align}\label{eqn:functional}
    \theta_0 =\E[\tilde m(W;h_0)]
\end{align}
where $h \mapsto \tilde m(W;h)$ is a known linear functional of $h$.
Given two $W$-measurable variables $X$ and $Z$,
the function $h_0$ is defined as a solution to a linear inverse problem:
\begin{align}\label{eqn:primal}
    \Tcal h=r_0,  
\end{align}
where $h_0$ lies in a closed linear sub-space $\bar{\Hcal}$ of the space $L_2(X)$ of square integrable functions of $X$,
$\Tcal: \bar{\Hcal} \mapsto \bar{\Qcal}$ is a projected conditional expectation operator, \ie, $\Tcal h=\Pi_{\bar{\Qcal}} \E[h(X)\mid Z=\cdot]$, where $\bar{\Qcal}$ is a closed linear sub-space of $L_2(Z)$ and $\Pi_{\bar{\Qcal}}$ denotes the mean-squared projection on the space $\bar{\Qcal}$, and $r_0\in \bar{\Qcal}$ is the Riesz representer of another known linear functional 
$q \mapsto m(W;q)$, \ie:
\begin{align}\label{eq:Reisz}
\forall q \in \bar{\Qcal}:~  \E[m(W;q)] = \E[r_0(Z)\, q(Z)].
\end{align}
The prototypical case is when $\bar{\Hcal}=L_2(X), \bar{\Qcal}=L_2(Z)$ and the moment $m$ is of the simple form $m(W;f)=Y f(Z)$ for some $W$-measurable variable $Y$, in which case $r_0=\E[Y\mid Z]$ and $h_0$ corresponds to the solution of a non-parametric instrumental variable (IV) regression problem, \ie, $\E[Y - h(X)\mid Z]=0$.

Despite the significance of this problem, asymptotically normal inference for the parameter of interest $\theta_0$ presents considerable challenges, particularly when the nuisance function $h_0$ is weakly identified. In particular, since $h_0$ is the solution to an inverse problem, many qualitative attributes of $h_0$ could be smoothened out or distorted by the linear operator $\Tcal$, to the point that we would need very many samples to recover these attributes well, or even to the point that these attributes are irrecoverable even in the limit of infinite samples, \ie, the function $h_0$ is not uniquely identified by \cref{eqn:primal}. These problems are typically referred to in the literature as the ill-posedness of the inverse problem \citep{carrasco2007linear,cavalier2011inverse}. A large line of work assumes unique identification in the limit and further imposes quantitative bounds on measures of ill-posedness, so as to establish estimation rates for the function $h_0$. However, it is known that the uniqueness assumption is easily violated in practical scenarios \citep{newey2003instrumental,andrews2005inference,santos2012inference,kallus2021causal}. 

We instead focus on the minimum norm solution to the inverse problem, which removes the need for a uniquely identified $h_0$; even when the inverse problem has multiple solutions, the minimum norm solution is necessarily unique. Most importantly, under reasonable assumptions, the parameter $\theta_0$ of interest is invariant to the chosen solution of the inverse problem and therefore focusing on the minimum norm solution $h_0$ is without loss of generality. More concretely, let 
$a_0 \in \bar{\Hcal}$ denote the Riesz representer of the linear functional $h\mapsto \E[\tilde{m}(W;h)]$:
\begin{align}
\forall h\in \bar{\Hcal}: \E[\tilde{m}(W;h)]=\E[a_0(X)\, h(X)].
\end{align}
Moreover, let $\Tcal^*:\bar{\Qcal}\mapsto \bar{\Hcal}$ denote the adjoint operator of $\Tcal$, which corresponds to the projected conditional expectation operator $\Tcal^* q=\Pi_{\bar{\Hcal}}\E[q(Z)\mid X=\cdot]$. As was already shown in prior work of \cite{severini2006some,severini2012efficiency,bennett2022inference}, if the dual inverse problem $\Tcal^* q = a_0$ admits any solution $q_0 \in \bar{\Qcal} \subseteq L_2(Z)$, then $\theta_0=\E[\tilde{m}(W;h_0)]$ takes the same value for any $h_0\in \bar{\Hcal}$
solving \cref{eqn:primal}, irrespective of what solution we use. 
Intuitively, existence of a solution $q_0$ is an assumption that the Riesz representer $a_0$ lies primarily on the higher order spectrum of the eigendecomposition of the operator $\Tcal$. Even though $h_0$ is not uniquely identified, the parameter $\theta_0$ is the projection of $h_0$ on the higher order spectrum, and this projection is uniquely identified. 

When such a solution $q_0$ exists, the parameter $\theta_0$ also admits a doubly robust representation:
\begin{align}
    \theta(h, q) := \E[\tilde m(W;h) + m(W;h)  - h(X)\, q(Z) ], 
\end{align}
which satisfies the mixed bias property:
\begin{align}
    \theta(h, q) - \theta_0 = \E[(q(Z) - q_0(Z))\, (h_0(X) - h(X))]
\end{align}
This formula lends itself to a natural estimation strategy: estimate $\hat{h}$ and $\hat{q}$ on a separate sample and then estimate $\theta_0$ in a plug-in manner, by taking the empirical analogue of the doubly robust representation formula:
\begin{align}\label{eqn:estimator}
    \hat{\theta} := \E_n[\tilde m(W;\hat h) + m(W;\hat h)  - \hat h(X) \hat q(Z)]. 
\end{align}
Prior work of \cite{chernozhukov2023simple, bennett2022inference} shows that this estimate is root-$n$ asymptotically normal when
\begin{align}
\Bcal_n := \min\Big\{ \|\Tcal (\hat h-h_0)\|_{L_2}\|\hat q-q_0\|_{L_2},~~\|\hat h-h_0\|_{L_2}\|\Tcal^{*}(\hat q-q_0)\|_{L_2} \Big\} = o_p(n^{-1/2}). 
\end{align}
and both $\|\hat{h}-h_0\|_{L_2}=o_p(1)$ and $\|\hat{q}-q_0\|_{L_2}=o_p(1)$. This observation implies that obtaining estimators with sufficient convergence guarantees for either $ \|\Tcal (\hat h-h_0)\|_{L_2}\|\hat q-q_0\|_{L_2}$ or $\|\hat h-h_0\|_{L_2}\|\Tcal^{*}(\hat q-q_0)\|_{L_2} $ is adequate for achieving asymptotically normal inference. Note that the 
estimation metric 
$\|\Tcal (\hat{h}-h_0)\|_{L_2}=\sqrt{\E\left[\left(\Pi_{\bar{\Qcal}}\E\left[\hat{h}(X)-h_0(X)\mid Z\right]\right)^2\right]}$,
which we refer to as the weak metric, can be much smaller than $\|\hat{h}-h_0\|_{L_2}=\sqrt{\E[(\hat{h}(X)-h_0(X))^2]}$, which we refer to as the strong metric. Estimation of $h_0$ with respect to the weak metric does not typically suffer from the ill-posedness of the inverse problem that defines $h_0$. Similarly for estimating $q_0$ with respect to its corresponding weak metric $\|\Tcal^{*}(\hat q-q_0)\|_{L_2}$. 
Thus for $\hat{\theta}$ to be root-$n$ asymptotically normal, it suffices to estimate only one of the two nuisance functions $h_0, q_0$ at a sufficiently fast rate with respect to the strong metric and the other with respect to the weak metric. Moreover, we always need to ensure consistency for both functions with respect to the strong metric, but without any rate.

\vspace{.5em}\noindent\textbf{Main contribution.} The main goal of our work is to establish novel estimators for $\hat{h}$ and $\hat{q}$ that allow for general non-linear function spaces and which provide guarantees on the quantity $\Bcal_n$ as a function of measures of ill-posedness of the primal and dual inverse problems that define $h_0$ and $q_0$, respectively, in the absence of unique identification.

\emph{Our main result is an ill-posedness doubly robust estimator:} we will give a single estimation algorithm for $\hat{h}$ and $\hat{q}$ such that if one of the two inverse problems is sufficiently well-posed then the parameter estimate $\hat{\theta}$ is root-$n$ asymptotically normal. \emph{Crucially the estimation algorithm does not need to know which inverse problem is the well-posed one.} Moreover, our estimation algorithm adapts to the degree of well-posedness of the most well-posed of the two inverse problems. For instance, as the largest of the two degrees goes to infinity, our requirements for asymptotic normality converge to the ones that correspond to the case when $h_0$ corresponds to the solution of a simple regression problem. \emph{This is the first such double robustness result, with respect to ill-posedness, in the literature.}

We measure the ill/well-posedness of the inverse problems using the \emph{source condition}.
Unlike other measures proposed in the literature, the source condition is an appropriate measure even in the absence of unique identification. To describe the source condition, let us restrict for the moment to linear operators that admit a countable singular value decomposition 
\begin{align}
    \Tcal h = \sum_{i=1}^\infty \sigma_i\, \langle h, v_i\rangle_{L_2}\, u_i ,
\end{align}
where $\sigma_1\geq\sigma_2\geq\ldots$ are the singular values;  
$\Tcal$ and form an orthonormal basis of $\bar{\Qcal}$, \ie, 
$\E[u_i(Z) u_j(Z)] = 1\{i=j\}$; and 
$v_i: X\to \R$ are the right eigenfunctions of $\Tcal$ and also form an orthonormal basis of $\bar{\Hcal}$, \ie, $\E[v_i(X) v_j(X)]=1\{i=j\}$.
The $\beta$-source condition on 
\cref{eqn:primal} states that the minimum norm solution $h_0$ is primarily supported on the 
lower part of the spectrum of the eigendecomposition of $\Tcal$:
\begin{align}
     \sum_{i=1}^\infty 1\{\sigma_i\neq 0\} \frac{\langle h_0, v_i\rangle_{L_2}^2 }{\sigma_i^{2\beta}} <\infty.
\end{align}
Note this implies that for any $m$, $\sum_{i=m}^\infty \langle h_0, v_i\rangle_{L_2}^2 \lesssim \sigma_m^{2\beta}$ (note also that the minimum norm solution has zero inner product with eigenfunctions for which $\sigma_i=0$). Thus the parameter $\beta$ in the source condition controls the amount of support that $h_0$ is allowed to have on the tail of the spectrum. As $\beta$ goes to infinity, the function $h_0$ behaves as being supported only on a finite set of eigenfunctions and the ill-posedness problem vanishes. More generally, $\beta$ controls the degree of ill-posedness and larger $\beta$ means that the inverse problem is more well-posed. 

The source condition is well defined even when the linear operator does not admit a singular value decomposition and, in its general form, requires the minimum norm solution $h_0$ to satisfy
\begin{align}
   \exists w_0\in \bar{\Hcal} : h_0 = (\Tcal^* \Tcal)^{\beta/2} w_0.
\end{align}

Let $\beta_h, \beta_q$ denote the degree of well-posedness of the primal and dual inverse problems that define $h_0$ and $q_0$ respectively, and let $\beta = \max\{\beta_h,\beta_q\}$, denote the largest of the two degrees, \ie, the degree of the most well-posed of the two inverse problems. 
Moreover, we will impose the inductive biases that minimum norm solutions to the inverse problems and regularized variants of them belong to the smaller function classes $\Hcal\subseteq \bar{\Hcal}, \Qcal\subseteq \bar{\Qcal}$.
Let $\delta_n$ denote the statistical complexity of appropriately defined function classes related to $\Hcal, \Qcal$ and function spaces $\Fcal, \Gcal$ that encompass their composition with the linear operators, \ie, $\bar{\Qcal} \supseteq \Fcal \supseteq \Tcal \circ (h_0 - \Hcal)=\{\Tcal (h_0 - h): h\in \Hcal\}$ and $\bar{\Hcal} \supseteq \Gcal \supseteq \Tcal^*\circ (q_0 - \Qcal) = \{\Tcal^* (q_0 - q): q\in \Qcal\}$. We will measure statistical complexity using the well-established notion of the critical radius, defined via the means of localized Rademacher complexities of the corresponding classes. 

Our main technical result is the development of an estimation algorithm for $\hat{h},\hat{q}$, such that the resulting parameter estimate $\hat{\theta}$ is root-$n$ asymptotically normal if:
\begin{align}
    \delta_n = o(n^{-\alpha}),\quad \alpha := \min\left( \frac{1 + \min(\beta,1) }{ 2 + 4 \min(\beta,1) }, \frac{1+\beta}{4 \beta} \right). 
\end{align}
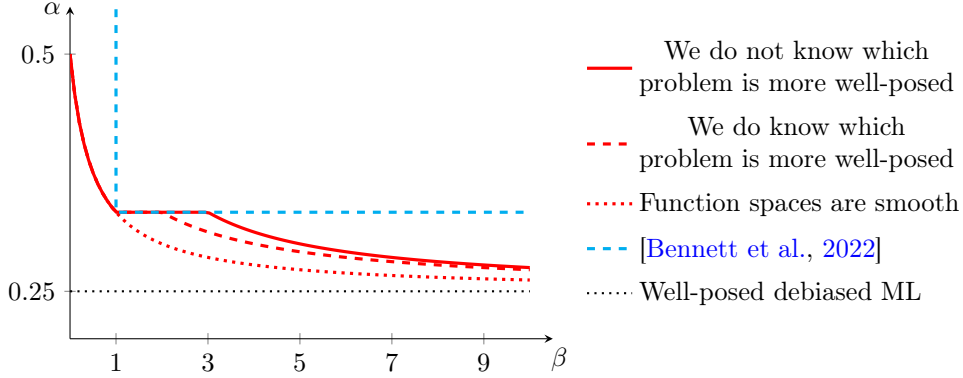
\begin{figure}[!t]
    \centering
    \begin{tikzpicture}
\begin{axis}[
    axis lines = left,
    xlabel = $\beta$,
    ylabel = {$\alpha$},
    legend style={at={(1.05,0.5)}, anchor=west, draw=none, row sep=0.15cm,cells={align=center}},
    ytick={0.25,0.5},
    xtick={1,3,...,9},
    ymin=0.2, ymax=0.55,
    xmin=0, xmax=10.5,
    no markers,
    legend cell align={left},
    width=8cm,
    height=6cm,
    xlabel near ticks,
    ylabel near ticks,
    xlabel style = {at={(axis description cs:1.05,0)},anchor=north east},
    ylabel style = {at={(axis description cs:0,0.95)},anchor=south east},
    ylabel style={rotate=-90}
]
\addplot[red, very thick, domain=0:10, samples=100] {min((1+min(x,1))/(2+4*min(x,1)),(1+x)/(4*x))};
\addlegendentry{We do not know which\\problem is more well-posed}
\addplot[red, dashed, very thick, domain=0:10, samples=100] {min((1+min(x,1))/(2+4*min(x,1)),(2+x)/(4+4*x))};
\addlegendentry{We do know which\\problem is more well-posed}
\addplot[red, dotted, very thick, domain=0:10, samples=100] {(1+x)/(2+4*x)};
\addlegendentry{Function spaces are smooth}
\addplot[cyan, dashed, very thick, domain=0:10, samples=100] coordinates {(0,1) (1,1) (1,1/3) (10,1/3)};
\addlegendentry{\citep{bennett2022inference}}
\addplot[black, dotted, thick, domain=0:10] {1/4};
\addlegendentry{Well-posed debiased ML}
\end{axis}
\end{tikzpicture}
    \caption{Required rate $\alpha$ in $\delta_n=o(n^{-\alpha})$ to achieve asymptotically normal inference, where $\delta_n$ is the critical radius of certain function classes. The x-axis corresponds to the degree of the source condition, and y-axis is a required rate for $\delta_n$ to achieve asymptotically normal inference.}
     \label{fig:required_rate}
\end{figure}

Notably, as $\beta\to \infty$, we get that we require $\delta_n = o(n^{-1/4})$, which is the requirement when $h_0$ is the solution to a regression, or equivalently a conditional expectation, problem \citep{chernozhukov2017double}. Moreover, for $\beta \in [1, 3]$, the requirement is $\delta_n=o(n^{-1/3})$, matching the prior work of \cite{bennett2022inference}, which applied only when $\beta=1$ and this degree of well-posedness was assumed to be satisfied by the inverse problem that defines $q_0$. Finally, even for severely ill-posed problems, where $\beta\geq \epsilon>0$ for some small $\epsilon$, the requirement is $\delta_n=o(n^{-1/2 + \kappa})$ for some $\kappa>0$, 
which is satisfied, for instance, for VC-subgraph classes.

If we know which of the two inverse problems is more well-posed, then we show that we can further weaken the requirement to:
\begin{align}
    \delta_n = o(n^{-\alpha}),\quad \alpha := \min\left( \frac{1 + \min(\beta,1) }{ 2 + 4 \min(\beta,1) }, \frac{2+\beta}{4 + 4 \beta} \right). 
\end{align}
Notably the loss due to not knowing which inverse problem is more well posed is minimal and primarily occurs when $\beta\in [2,3]$. Importantly, there is no loss for moderately ill-posed problems when $\beta\leq 1$, which is arguably the most practically relevant case. Moreover, if we make the further assumption that the function spaces are smooth enough that the projected $L_2$ norm is related to the projected $L_\infty$ norm, \ie, $\|T h\|_{L_\infty}=O\left(\|Th\|_{L_2}^\gamma\right)$, and similarly for $q$, then this loss can be further ameliorated. For instance, as $\gamma\to 1$ (a property satisfied for instance by Reproducing Kernel Hilbert Spaces with an exponential eigendecay, which is the case for the Gaussian kernel), then there is no loss to not knowing which inverse function is more well-posed and the requirement is always of the even weaker form:
\begin{align}
    \delta_n = o(n^{-\alpha}),\quad \alpha := \frac{1 + \beta }{ 2 + 4 \beta}
\end{align}

\vspace{.5em}\noindent\textbf{Main techniques.} Our result is enabled by several novel contributions of independent interest. Our first goal is the development of estimation algorithms for a function defined via a linear inverse problem that satisfies a known $\beta$-source condition. Such algorithms can be applied both for the primal and for the dual inverse problems. We describe our results here in the context of the primal inverse problem $\Tcal h=r_0$. 

Our overall goal is an estimation algorithm that produces an estimate $\hat{h}$, such that irrespective of whether the source condition holds, it guarantees a fast convergence rate of $O(\delta_n^2)$, with respect to the weak metric. Moreover, when the source condition does hold it also guarantees a statistical rate of $O(\delta_n^{\kappa(\beta_h)})$, for some exponent $\kappa(\beta_h)$ with respect to the strong metric. Note that if we manage to construct such an estimation algorithm, then if we apply this algorithm for both the primal and the dual inverse problems, then we will be guaranteeing that $\Bcal_n=O\left(\delta_n^{2 + \kappa(\beta)}\right)$, where $\beta=\max\{\beta_h,\beta_q\}$, which would then lead to the required condition on $\delta_n$.

Achieving a fast learning rate with respect to the weak metric has already been established in prior work of \cite{dikkala2020minimax}, with the use of an adversarial estimation strategy, which was further extended in \citep{bennett2022inference} to ensure consistency with respect to the strong metric, via the means of Tikhonov regularization, for the following estimator:
\begin{align}
    \hat{h} = \argmin_{h\in \Hcal} \max_{f\in \Fcal} \E_n\left[2\,m(W;f) - h(X)\, f(Z) - f(Z)^2\right] + \lambda \E_n\left[h(X)^2\right]
\end{align}
Moreover, prior work of \cite{LiaoLuofeng2020PENE} established strong metric rates for this Tikhonov regularized adversarial estimator, for smooth function classes with neural network function approximation. However, the rate in \citep{LiaoLuofeng2020PENE} is sub-optimal, does not adapt to the critical radius of arbitrary function spaces, and does not adapt to large values of $\beta$ (only to $\beta\leq 1$). Our first result is a fast strong metric rate result for the Tikhonov regularized estimator for $\beta \leq 1$. Our second result is a fast rate result for an iterated version of the Tikhonov regularized estimator, that uses prior iteration estimates to center the regularization appropriately and which adapts to large values of $\beta$, leading to strong metric rates of $O(\delta_n^2)$ as $\beta \to \infty$. These two results are novel in the literature on estimation of linear inverse problems under a source condition and are of independent interest.

Finally, we show that the desired simultaneous guarantee can be ensured via a constrained Tikhonov regularized adversarial estimator. In particular, our estimator first solves the un-regularized objective to find a solution that guarantees a fast weak metric rate. Subsequently, it solves the regularized objective within the sub-space of solutions that also achieve a small un-regularized risk, compared to the un-regularized solution. We show that this estimator simultaneously enjoys both guarantees: a weak metric rate of $\delta_n^2$, without requiring a $\beta$-source condition and a strong metric rate of $\delta_n^{2\max\left\{\frac{\min(\beta,1)}{1+\min(\beta,1)}, \frac{\beta-1}{\beta+1}\right\}}$, when the $\beta$-source condition holds. This theorem enables our main double robustness result.

%

%

%

%
%
%

%
%
%
%
%
%
%
%
%
%
%

%
%
%
%
%

%
%
%
%
%

%

%

\section{Related Work}\label{sec:related}

We first discuss related work specifically focusing on nonparametric IV regression functions, \ie, solutions to $\E[Y-h
(X)\mid Z]=0$. Later, we delve into related work that focuses on estimating functionals $\theta_0$ of nuisance functions that are defined as linear inverse problems, including nonparametric IV regression functions.

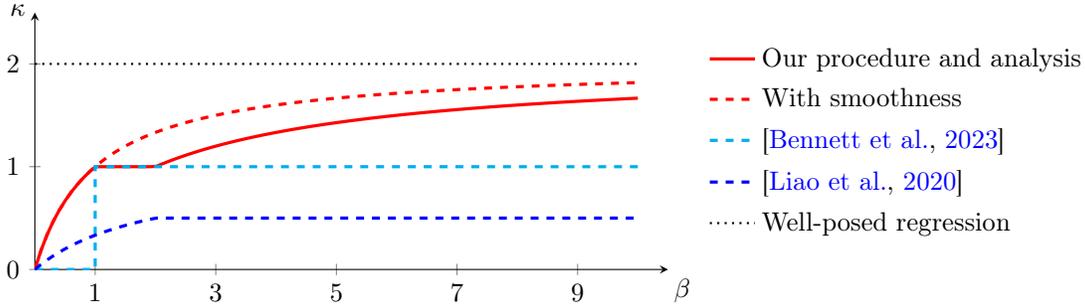
\begin{figure}[!t]
\centering
\begin{tikzpicture}
\begin{axis}[
    axis lines = left,
    xlabel = $\beta$,
    ylabel = {$\kappa$},
    legend style={at={(1.05,0.5)}, anchor=west, draw=none, row sep=0.1cm},
    ytick={0,1,2},
    xtick={1,3,...,9},
    ymin=0, ymax=2.5,
    xmin=0, xmax=10.5,
    no markers,
    legend cell align={left},
    width=10cm,
    height=5cm,
    xlabel near ticks,
    ylabel near ticks,
    xlabel style = {at={(axis description cs:1.05,0)},anchor=north east},
    ylabel style = {at={(axis description cs:0,0.95)},anchor=south east},
    ylabel style={rotate=-90}
]
\addplot[red, very thick, domain=0:10, samples=100] {2 * max(min(x,1)/(1+min(x,1)), x/(2+x))};
\addlegendentry{Our procedure and analysis}
\addplot[red, dashed, very thick, domain=0:10, samples=100] {2 * max(min(x,1)/(1+min(x,1)), x/(1+x))};
\addlegendentry{With smoothness}
\addplot[cyan, dashed, very thick, domain=0:10, samples=100] coordinates {(0,0) (1,0) (1,1) (10,1)};
\addlegendentry{\citep{bennett2023minimax}}
\addplot[blue, dashed, very thick, domain=0:10, samples=100] {min(x,2)/(2+min(x,2))};
\addlegendentry{\citep{LiaoLuofeng2020PENE}}
\addplot[black, dotted, thick, domain=0:10] {2};
\addlegendentry{Well-posed regression}
\end{axis}
\end{tikzpicture}
\caption{Exponent $\kappa$ in the rate $\|\hat{h}-h_0\|^2_2\sim \delta_n^{\kappa}$ as a function of the degree $\beta$ of the source condition in our work and related works \citep{LiaoLuofeng2020PENE,bennett2023minimax}. Here, $\delta_n$ represents a critical radius of a specific function class.}
\label{fig:comparison}
\end{figure}

\noindent\textbf{Nonparametric IV regression.}
Instrumental variable estimation has garnered significant interest as a particular subset within the realm of inverse problems, as exemplified by the comprehensive investigations \citep{carrasco2007linear,cavalier2011inverse, newey2013nonparametric}.
Nonparametric instrumental variable estimation encounters considerable challenges due to its ill-posed nature, even when the operator $\mathcal{T}$ and response $r_0$ are known. The ill-posedness is often characterized by the presence of one or more of the following aspects: (1) the absence of solutions, (2) the existence of multiple solutions, and (3) the discontinuity of the pseudo-inverse of $\Tcal$. To tackle these challenges, various regularization techniques have been proposed, such as imposing compactness on the solution space \citep{newey2003instrumental} and employing Tikhonov regularization \citep{carrasco2007linear}. In practical settings where $\mathcal{T}$ and $r_0$ are unknown, a range of estimators has been proposed in the literature, including series-based estimators \citep{ai2003efficient,hall2005nonparametric,blundell2007semi,chen2011rate,darolles2011nonparametric,chen2012estimation,florens2011identification,chen2021robust}, kernel-based estimators \citep{hall2005nonparametric,horowitz2007asymptotic}, RKHS-based estimators \citep{singh2019kernel,muandet2020dual,bennett2020variational} and high-dimensional linear estimators under sparsity \citep{gautier2011high,gautier2018high,gautier2022fast}. 

Recently, there has been an increasing interest in applying general function approximation techniques, such as deep neural networks and random forests, to instrumental variable problems in a unified manner \citep{dikkala2020minimax,lewis2018adversarial,bennett2019deep,zhang2020maximum,dikkala2020minimax,bennett2020variational}. However, the guarantees of most current methods remain unclear when solutions are not unique. Notable exceptions include the works of \citep{LiaoLuofeng2020PENE,bennett2023minimax},which provide finite-sample convergence rate guarantees even when solutions may not be unique.

In the closely related work of \citep{LiaoLuofeng2020PENE}, they establish $L_2$ convergence by connecting minimax optimization with Tikhonov regularization under the assumption of the source condition. Specifically, when the number of iterations is limited to one in our method, their method coincides with ours in solving inverse problems. However, our paper presents two important contributions beyond their work. Firstly, we introduce a new iterative procedure that achieves a fast convergence rate under high-order source conditions with $\beta \geq 2$. Secondly, even when the number of iterations is limited to one, our paper achieves a faster convergence rate than theirs due to improved analysis. We note that their paper also makes its own contribution by specifically considering scenarios where function classes are neural networks and providing a theoretical analysis that takes into account the optimization procedure.

Another closely related work \citep{bennett2023minimax} proposes a method that treats IV regression as a constrained optimization problem. While their method does not require the ``closedness assumption,''  which implies smoothness of the operator $\mathcal{T}$, compared to our work, the guaranteed convergence rate in their paper is slower than ours because their method cannot effectively exploit potentially high-order source conditions with $\beta \geq 2$. Additionally, their method does not provide any guarantees under the source condition with $\beta \geq 2$.

We note that there are a number of alternative approaches for integrating machine learning into instrumental variable estimation \citep{hartford2017deep,yu2018deep,xu2020learning,liu2020deep,kato2021learning,lu2021machine}. However, to the best of our knowledge, these approaches do not offer an $L_2$ convergence rate guarantee in the absence of the assumption of uniqueness.

%
%


\section{Problem Statement and Preliminaries}

We are given access to a set of independent and identically distributed observations $\{X_i,Z_i,W_i\}_{i=1}^n$ drawn from the distribution of the random variables $X, Z, W$. Our ultimate objective is to estimate the parameter $\theta_0 = \E[\tilde{m}(W; h_0)]$, as presented in \eqref{eqn:functional}, and
construct a valid confidence interval around it. 

Throughout this work, we assume there exists a solution to the linear inverse problem given by \cref{eqn:primal,eq:Reisz}. 
(Numbered assumptions are assumed to hold throughout the paper.) 
\begin{assumption}[Primal solution exists]\label{assum:existence}
   We have  $r_0 \in \mathcal{R}(\Tcal)$, where $\mathcal{R}(\Tcal):= \{\Tcal h: h\in \bar{\Hcal}\}$. 
\end{assumption}
Moreover, it will be convenient to express the constraints that identify $h_0$ in a combined manner, which can be derived by simple algebra from \cref{eqn:primal,eq:Reisz} 
and the properties of mean-squared projections onto closed linear spaces:\footnote{For any $q\in \bar{\Qcal}$, $r\in L_2(Z)$ it follows from properties of projections on closed linear spaces that $\langle r, q\rangle_{L_2(Z)} = \langle \Pi_{\bar{\Qcal}}r, q\rangle_{L_2(Z)}$. Hence: $\E[h_0(X) q(Z)] = \E[\E[h_0(X)\mid Z]\, q(Z)] = \E[(\Tcal h_0)(Z)\, q(Z)] = \E[r_0(Z) q(Z)] = \E[m(W;q)]$.}
\begin{align}\label{eqn:moment-restrictions}
    \forall q\in \bar{\Qcal}: \E[m(W; q) - h_0(X)\, q(Z)] = 0.
\end{align}

In general, the solution mentioned above may not be unique. For this reason we aim to estimate the least $L_2$-norm solution, \ie, we define $h_0$ as 
\begin{align}
    h_0 = \argmin_{h: \Tcal h=r_0} \|h\|_{L_2}
\end{align}
This least norm solution always exists uniquely, as shown in Lemma 1 of \cite{bennett2023minimax}, and is frequently employed as a target in the literature when solutions are not unique \citep{santos2011instrumental,florens2011identification}. Notably, as we elaborate in the next section, for many linear functionals, the specific choice of the solution to the linear inverse problem is irrelevant and all solutions lead to the same value for the parameter $\theta_0$.

Our setting encompasses many well-studied problems in econometrics and statistics. We present here two illustrative examples. 
Other examples that fall in our framework include partially linear IV and proximal causal inference models,
missing-not-at-random data with shadow variables, and offline policy evaluation in partially observable MDPs.

\begin{example}[Proximal Causal Inference \citep{cui2020semiparametric}]
In proximal causal inference, we aim to estimate the average treatment effect $\E[Y(1)-Y(0)]$ in the presence of unmeasured confounders. Given proxy variables $Z,Q$ that satisfy certain conditions in \citep{cui2020semiparametric}, and observed treatment $D$, the target is expressed as 
\begin{align}
\theta_0 = \E[h_0(X,Q, 1) - h_0(X,Q,0)],\quad 
\E[Y - h_0(X,Q,D) \mid X, Z, D] = 0. 
\end{align}
A function $h_0$ is often referred to as the outcome bridge function.
In this case, $\tilde{m}(X,Q;h) = h(X,Q,1) - h(X,Q,0)$, $m(W;q)=Y\, q(X,Z,D)$ and the operator $\Tcal$ maps any function $h(X,Q,D) \in L_2(X,Q,D)$ to $\E[h(X,Q,D)|X,Z,D] \in L_2(X,Z,D)$.
Furthermore, the Riesz representer of $\tilde m(W;h)$ is 
\begin{align}
    \alpha_0(X,Q,D) = \frac{D}{P(D=1 | X,Q)} - \frac{1-D}{P(D=0 | X,Q)}
\end{align}
since $\theta_ 0 = \E[\alpha_0(X,Q,D) h_0(X,Q,D)]$. 
\end{example}

\begin{example}[Average Price Elasticity] In many demand estimation problems, we want to use cost-shifters $Z$ (variables that affect the cost of a product, and hence the price, but not the demand), so as to estimate the demand $Y$ of some product as a function of price $D$, conditional on market characteristics $X$. In this case the problem typically boils down to a non-parametric instrumental variable regression, where $(Z,X)$ is the instrument, $(D,X)$ is the treatment, $Y$ is the outcome, and $h_0$ is the demand curve. Since learning the whole demand curve can be statistically challenging, many times it might suffice for policy purposes to simply learn the average price elasticity of demand. One way to formalize the average price elasticity is through the average derivative of the demand curve:
\begin{align}
    \theta_0 :=~& \E[\partial_D h_0(X, D)] &
    \E[Y - h_0(X, D)\mid Z, D] =~& 0
\end{align}
In this case, $\tilde{m}(W;h) = \partial_D h_0(X, D)$ and $m(W;q)=Y\, q(Z)$. Furthermore the Riesz representer of $\tilde{m}(W;h)$ is $a_0(X, D) = \partial_{D} \log(p(D\mid X))$, where $p(D\mid X)$ is the conditional density of the price $D$ given $X$.
\end{example}

\vspace{.5em}\noindent\textbf{Notation and preliminary definitions.} Before delving into the main technical exposition we need to introduce some technical notation. \emph{Throughout the paper, whenever we use a generic norm of a function $\|h\|$, we will be referring to the $L_2$-norm with respect to the distribution of the input of the function, \ie,}
\begin{align}
\forall h\in L_2(X): \|h\| =~& \|h\|_{L_2} = \sqrt{\E[h(X)^2]} &
\forall f\in L_2(Z): \|f\| =~& \|f\|_{L_2} = \sqrt{\E[f(Z)^2]}
\end{align}
We will also be using the shorthand notation $\E_n[\cdot]$ for the empirical average, \eg, $\E_n[X] = \frac{1}{n}\sum_{i=1}^n X_i$. For any set $A$, we denote the closure of $A$ by $\clos(A)$.
For any function space $\Fcal$, containing functions that are uniformly and absolutely bounded by $1$, we will be using the critical radius as the measure of statistical complexity (c.f. \citep{wainwright2019high} for a more detailed exposition). To define the critical radius we first define the localized Rademacher complexity:
\begin{align}
    \Rcal(\Fcal, \delta) = \frac1{2^n}\sum_{\epsilon\in\{-1,1\}^n}\E\left[\sup_{f\in \Fcal: \|f\|\leq \delta} \frac{1}{n}\sum_{i=1}^n \epsilon_i f(X_i)\right].
\end{align}
The star hull of a function space is defined as $\shull(\Fcal)=\{\gamma f: f\in \Fcal, \gamma\in [0,1]\}$. The critical radius $\delta_n$ of $\shull(\Fcal)$ 
is the smallest positive solution to the inequality:
\begin{align}
    \Rcal(\shull(\Fcal), \delta) \leq \delta^2
\end{align}
Throughout we will be assuming a uniform absolute bound of $1$ for all random variables and functions. This can be lifted to any finite bound $b$ by rescaling.
\begin{assumption}[Uniform absolute bound]\label{ass:ubound}
    All random variables and random functions are uniformly absolutely bounded by $1$.
\end{assumption}

\section{Debiased Machine Learning Inference for Functionals}

We begin by noting that any linear functional $\tilde{m}$ of $h_0$ admits a doubly robust representation:
\begin{align}\label{eqn:doubly_robust}
    \theta_0 = \theta(h_0,q_0),\quad \theta(h,q) := \E[\tilde{m}(W;h) + m(W;q) - q(Z)\, h(X)],
\end{align}
where $q_0$ solves a dual inverse problem of the same nature as $h_0$, but with respect to functional $\tilde{m}$ instead of $m$. More specifically:
\begin{align}\label{eqn:IV-q}
    \Tcal^* q_0 =~& a_0 &
    \forall h\in \Hcal: \E[\tilde{m}(W;h)] = \E[a_0(X)\, h(X)],
\end{align}
where $\Tcal^*:\bar{\Qcal}\mapsto \bar{\Hcal}$ is the adjoint operator to $\Tcal:\bar{\Hcal}\mapsto \bar{\Qcal}$. The adjoint operator is defined as:
\begin{align}
    \Tcal^* q = \Pi_{\bar{\Hcal}} \E[q(Z)\mid X=\cdot] = \argmin_{h\in \bar{\Hcal}} \E\left[\left(h(X) - \E[q(Z)\mid X]\right)^2\right].
\end{align}
If $\bar{\Hcal}= L_2(\mcX)$, then $\Tcal^*$ simplifies to a conditional expectation $T^*q = \E[q(Z)\mid X=\cdot]$. Note that the conditions that define $q_0$ can be expressed in a combined manner, similar to Equation~\eqref{eqn:moment-restrictions}:
\begin{align}\label{eqn:dual-combined}
    \forall h\in \bar{\Hcal}: \E[\tilde{m}(W;h) - q_0(Z)\, h(X)] = 0.
\end{align}

We assume the existence of a solution $q_0\in \bar{\Qcal}$ to the inverse problem defined in Equation~\eqref{eqn:IV-q}.
If the inverse problem in Equation~\eqref{eqn:IV-q} has multiple solutions, then we will again denote by $q_0$ the minimum $L_2$-norm solution to the inverse problem.
\begin{assumption}[Dual solution exists]
    We have $a_0  \in \Rcal(\Tcal^*)$, with $\mathcal{R}(\Tcal^*):= \{\Tcal^* q: q\in \bar{\Qcal}\}$.  
\end{assumption}

The reader might wonder why we need this assumption to estimate $\theta_0$ on top of the existence of solutions in the primal inverse problem. As was shown in \citep{severini2012efficiency}, when 
$\bar{\Hcal}=L_2(X)$ and $\bar{\Qcal}=L_2(Z)$,
the assumption
$a_0  \in \Rcal(\Tcal^*)$ is a necessary condition for the $\sqrt{n}$-estimability of the parameter $\theta_0$. In this sense the assumption is unavoidable for root-$n$ asymptotic normality. 

Now, we are ready to state the mixed bias property of \eqref{eqn:doubly_robust}. 
\begin{lemma}\label{lem:doubly_robust}
    $\theta(h,q)$ defined in \cref{eqn:doubly_robust} satisfies the mixed bias (or, double robustness) property 
    that for all $h\in \bar{\Hcal}$ and $q\in \bar{\Qcal}$:
    \begin{align}
        \theta(h,q) - \theta_0 = \E[(q_0(Z) - q(Z))\, (h(X) - h_0(X))]
    \end{align}
\end{lemma}

Based on this crucial lemma, we can then apply the general machinery of Neyman orthogonality and debiased machine learning \citep{chernozhukov2017double} to arrive at a corollary that shows how and when one can deduce root-$n$ consistency and asymptotic normality of the estimate $\hat{\theta}$ of $\theta_0$ presented in Equation~\eqref{eqn:estimator} (see also \citep{chernozhukov2023simple} for a finite sample version in a slightly simpler setting):
\begin{corollary}[DML Inference for Functionals of Inverse Problems]\label{cor:functionals-endo}
    Assume that estimates $\hat{h},\hat{q}$ of $h_0,q_0$ are estimated on a separate sample with
    \begin{align}
        \sqrt{n} \E\left[(\hat{q}(Z) - q_0(Z))\, (\hat{h}(X) - h_0(X))\right] = o_p(1),
    \end{align}
    $\|\hat{h}-h_0\|=o_p(1)$, and $\|\hat{q}-q_0\|=o_p(1)$.
Assume $\hat{\theta}$ is constructed as in \cref{eqn:estimator}
    and that the following mean-squared-continuity property is satisfied for $\tilde{m}$:
    \begin{align}
        \E[(\tilde{m}(W;\hat{h})-\tilde{m}(W;h_0))^2] = O\left(\|\hat{h}-h_0\|^{\iota}\right)\quad\text{for some $\iota>0$}.
    \end{align}
     Then:
    \begin{align}
        \sqrt{n} (\hat{\theta} - \theta_0) = \frac{1}{\sqrt{n}} \sum_{i=1}^n \rho_0(W_i) + o_p(1)
    \end{align}
    with $\rho_0(W) = \tilde{m}(W;h_0) + m(W;q_0) - q_0(Z)\,h_0(X) - \theta_0$. Thus, the random variable $\sqrt{n}\,(\hat{\theta}-\theta_0)$ converges in distribution to a normal $N\left(0, \E[\rho_0(W)^2]\right)$. Moreover, if we let $\hat{\sigma} = \E_n\left[\left(\tilde{m}(W;\hat{h}) + m(W;\hat{q}) - \hat{q}(Z)\, \hat{h}(X) - \hat{\theta}\right)^2\right]$,
     we can construct an asymptotically valid 95\% confidence interval using:
    \begin{align}
        \theta_0 \in [\hat{\theta} \pm 1.96\cdot \hat{\sigma}/\sqrt{n}]
    \end{align}
\end{corollary}

\begin{remark}[Uniqueness of Parameter under Non-Uniqueness of Functions]
    We note that when \cref{eqn:IV-q} admits a solution $q_0$, then the parameter $\theta_0$ is uniquely identified, even when the primary inverse problem in Equation~\eqref{eqn:moment-restrictions} admits many solutions. Let $h_1, h_2$ be any two solutions to \cref{eqn:moment-restrictions} and $q_0$ any solution to \cref{eqn:IV-q} and let $\theta_1$ and $\theta_2$ the parameters that correspond to these two solutions. Then note that:
    \begin{align}
        \theta_1 := \E[\tilde{m}(W; h_1)] =~& \E[q_0(Z)h_1(X)] \tag{by Equation~\eqref{eqn:dual-combined} for $h=h_1$}\\
        =~& \E[m(W;q_0)] \tag{$h_1$ satisfies Equation~\eqref{eqn:moment-restrictions} for $q=q_0$}\\
        =~& \E[q_0(Z) h_2(Z)] \tag{$h_2$ satisfies Equation~\eqref{eqn:moment-restrictions} for $q=q_0$}\\
        =~& \E[\tilde{m}(W; h_2)] =: \theta_2  \tag{by Equation~\eqref{eqn:dual-combined} for $h=h_2$}
    \end{align}
    Thus the linear functional parameter $\theta_0$ is the same, whether we calculate it using $h_1$ or $h_2$. Thus for functional estimation, it does not matter that our estimate is converging to the minimum norm solution $h_0$. Under the existence of a solution $q_0$ to \cref{eqn:IV-q}, any other solution to the inverse problem in \cref{eqn:moment-restrictions} would have identified the exact same parameter $\theta_0$.
\end{remark}

If we want to apply the latter corollary, it remains to show how we can estimate $\hat{h},\hat{q}$ in a manner such that:
\begin{align}\label{eqn:mixed-bias-req}
    \sqrt{n} \E\left[(\hat{q}(Z) - q_0(Z))\, (\hat{h}(X) - h_0(X))\right] = o_p(1)
\end{align}
and such that $\|\hat{h}-h_0\|=o_p(1)$ and $\|\hat{q}-h_0\|=o_p(1)$. 
By the definition of $\Tcal$, applying a tower law, a Cauchy-Schwarz inequality and orthogonality of projections on linear spaces, we have
\begin{align}
    \E\left[(\hat{q}(Z) - q_0(Z))\, (\hat{h}(X) - h_0(X)) \right] 
    &\leq \|\hat{q}-q_0\|\, \|\Tcal(\hat{h}-h_0)\|\,.
\end{align}
Similarly, by the definition of $\Tcal^*$ and the orthogonality of projections, we also have that:
\begin{align}
\E\left[(\hat{q}(Z) - q_0(Z))\, (\hat{h}(X) - h_0(X)) \right] 
&\leq \|\Tcal^*(\hat{q}-q_0)\|\, \|\hat{h}-h_0\|\,,
\end{align}
and therefore a sufficient condition for Equation~\eqref{eqn:mixed-bias-req} to hold is:
\begin{align}
    \sqrt{n} \min\left\{\|\hat{q}-q_0\|\, \|\Tcal(\hat{h}-h_0)\|, \|\Tcal^*(\hat{q}-q_0)\|\, \|\hat{h}-h_0\|\right\}  = o_p(1) \,.
\end{align}

Note that both $\hat{h}$ and $\hat{q}$ are instances of the same estimation problem. In particular, both statistical problems can be defined as finding a function $h\in \bar{\Hcal}$ that satisfies a linear inverse problem $\Tcal h=r_0$, where $r_0$ is the Riesz reprsenter of a linear functional and $\Tcal$ is a projected conditional expectation operator. Thus we will describe how one can solve the primal problem of estimating $\hat{h}$ and the exact same analysis can be applied to the dual problem where the operator $\Tcal$ is replaced by the dual $\Tcal^*$, the function space $\bar{\Hcal}$ is replaced by $\bar{\Qcal}$, and the linear functional $m(W;q)$ is replaced by the linear functional $\tilde{m}(W;h)$.

We will provide estimation rates as a function of inductive biases further imposed on the solutions $h_0$ and $q_0$. In particular, we will consider the following realizability assumptions:
\begin{assumption}[Realizability and Closedness]\label{ass:realizable}
    For some known convex function spaces $\Hcal\subseteq \bar{\Hcal}$ and $\Qcal\subseteq \bar{\Qcal}$, we assume $h_0 \in \Hcal$ and $q_0\in \Qcal$. Moreover, for some known function spaces $\Fcal\subseteq \bar{\Qcal}$ and $\Gcal\subseteq \bar{\Hcal}$, we have $\Tcal \circ (h_0 - \Hcal) :=\{\Tcal (h_0 - h): h\in \Hcal\}\subseteq \Fcal$ and $\Tcal^* \circ (q_0 - \Qcal) := \{\Tcal^* (q_0 - q): q\in \Qcal\}\subseteq \Gcal$.
\end{assumption}
The function classes $\Hcal, \Qcal, \Fcal, \Gcal$, are meant to be classes of bounded statistical complexity, such as for instance norm-constrained high-dimensional linear models, Reproducing Kernel Hilbert Spaces, or neural network classes. This assumption can be relaxed to allow for approximation errors in all inclusion statements, at the expense of additive such error bounds in all our theorems. We omit such an extension for simplicity of exposition.

\section{Linear Inverse Problems and the Source Condition}

Given access to $n$ samples $\{X_i, Z_i, W_i\}_{i=1}^n$, our objective is to find an estimator $\hat h$ that converges to the minimum norm solution $h_0$ to Equation~\eqref{eqn:primal}. Our goal is to provide estimation rates with respect to both the strong and weak metrics defined below:
\begin{align}
    \|\hat{h}-h_0\|^2 =~& \E[(\hat{h}(X) - h_0(X))^2] \tag{strong metric}\\
    \|\Tcal(\hat{h}-h_0)\|^2 =~& 
    \E[\Pi_{\bar{\Qcal}}\E[\hat{h}(X) - h_0(X)\mid Z]^2]
    \tag{weak metric}
\end{align}
The weak metric is ``weak'' in the sense that it is smaller than the strong metric and is a pseudo metric. Specifically, $\|\Tcal(\hat{h}-h_0)\|^2= \|\Tcal(\hat{h}-h'_0)\|^2$ whenever $Th_0'=r_0$ even if $h'_0\neq h_0$. 

Within the literature on nonparametric instrumental variable (IV) regression, a commonly employed approach focuses on optimizing empirical versions of the weak metric criterion, also known as the minimum distance criterion. This optimization process is carried out over simple hypothesis spaces denoted as $\Hcal_n$ of increasing complexity. These hypothesis spaces, often referred to as sieves, aim to approximate the function $h_0$ while enforcing some uniform control over the ratio between the strong and weak metrics across the entire class, also known as the measure of ill-posedness of the inverse problems:
\begin{align}
    \sup_{h \in \Hcal_n}\frac{\|(h_0-h)\|^2 }{\|\Tcal(h_0-h)\|^2}. 
\end{align}
This measure, as indicated in previous works \citep{dikkala2020minimax,chen2012estimation}, reflects the degree of ill-posedness in inverse problems. However, in cases where unique identification is absent this measure can easily diverge to infinity. If the sieves $\Hcal_n$ are taken so that eventually they uniformly approximate the function space $\Hcal$, then note that there will exist a different solution $h_0'\in \Hcal$ to the inverse problem, that is $\epsilon$ approximated by a function $h_{\epsilon}$ in $\Hcal_n$. Thus we will have that $\|h_0-h_0'\| \to \gamma > 0$ and $\|\Tcal(h_0-h_{\epsilon})\|\leq \epsilon + \|\Tcal(h_0-h_0')\| = \epsilon\to 0$.


An alternative approach that remains effective even in the absence of the uniqueness assumption involves imposing the $\beta$-source condition on the minimum norm solution $h_0$. Under the $\beta$-source condition, it becomes possible to achieve desirable rates on the strong metric by introducing an additional $L_2$ regularization penalty to the minimum distance criterion. This regularization technique is also known as Tikhonov regularization \citep{Carrasco2007}. In this study, we adopt the latter approach.

\vspace{.5em}\noindent\textbf{The Source Condition.} We begin by introducing the $\beta$-source condition, which is commonly used in the literature on inverse problems \citep{Carrasco2007} and captures mathematically how strongly the function $h_0$ is identified by the data that we observed.

\begin{definition}[$\beta$-source condition]\label{ass:source-primal}
There exists $w_0 \in \bar{\Hcal}$ such that $h_0$ satisfies $$h_0 = (\Tcal^* \Tcal)^{\beta/2} w_0.$$
\end{definition}


Essentially, this assumption is claiming 
\begin{align}
    h_0 \in \Rcal( (\Tcal^* \Tcal)^{\beta/2}),\,\,\text{or equivalently,}\,\,  r_0 \in \Rcal( (\Tcal \Tcal^*)^{(\beta+1)/2}). 
\end{align}
Intuitively, when the parameter $\beta$ is large, the function $h_0$ exhibits greater smoothness, in the sense  that the $L_2$-inner product of $h_0$ with eigenfunctions that have smaller eigenvalues relative to the operator $\Tcal$ tend to decay faster. A more concrete interpretation of the assumption when
the operator $\Tcal$ is compact and admits a singular value decomposition is given in \cref{sec: intro}.

\section{Tikhonov Regularized Adversarial Estimator}


Inspired by the combined moment constraints in Equation~\eqref{eqn:moment-restrictions} (and similarly Equation~\eqref{eqn:dual-combined} for $q_0$), we will consider an adversarial population criterion for the estimation of $h_0$:
\begin{align}
    L(h) := \max_{f\in \Fcal} \E[2\,(m(W;f) - h(X)\, f(Z)) - f(Z)^2]
\end{align}
where $\Fcal \subseteq \bar{\Qcal}$, as defined in Assumption~\ref{ass:realizable}. By the definition of the function $h_0$, we have $\E[m(W;f)]=\E[h_0(X)\, f(Z)]$. Hence, the above population criterion is equivalent to:
\begin{align}
    L(h) = \max_{f\in \Fcal} \E[2\,(h_0(X) - h(X))\, f(Z)) - f(Z)^2].
\end{align}
Moreover, by an application of the tower law of expectations and the properties of projections onto the closed linear space $\bar{\Qcal}$ and since $\Fcal\subseteq \bar{\Qcal}$, the criterion is also equivalent to:
\begin{align}
    L(h) =~& \max_{f\in \Fcal} 2\langle \E[h_0(X) - h(X)\mid Z=\cdot], f\rangle_{L_2(Z)} - \|f\|^2 \tag{tower law}\\
    =~& \max_{f\in \Fcal} 2\langle \Pi_{\bar{\Qcal}}\E[h_0(X) - h(X)\mid Z=\cdot], f\rangle_{L_2(Z)} - \|f\|^2 \tag{$f\in \bar{\Qcal}$ and $\bar{\Qcal}$ closed linear}\\
    =~& \max_{f\in \Fcal} 2\langle \Tcal (h_0 - h), f\rangle_{L_2(Z)} - \|f\|^2 \tag{definition of $\Tcal$}
\end{align}
Since $\Fcal$ contains functions of the form $\Tcal (h_0 - h)$ for any $h\in \Hcal$ (with $\Hcal$ as in Assumption~\ref{ass:realizable}):
\begin{align}
    f_h := \argmax_{f\in \Fcal} \E[2(m(W;f) - h(X)\, f(Z)) - f(Z)^2] = \Tcal (h_0 - h)
\end{align}
and the population criterion is equivalent to the weak metric, for any $h\in \Hcal$:
\begin{align}
    \begin{aligned}
        L(h) =~& 2 \langle \Tcal (h_0 -h), f_h\rangle_{L_2(Z)} - \|f_h\|^2
    ~=~ \|\Tcal(h_0-h)\|^2
    \end{aligned}\label{eqn:equivalence-adv-weak}
\end{align}
In other words, by minimizing the adversarial population criterion, we are essentially minimizing the weak metric distance to $h_0$.

\vspace{.5em}\noindent\textbf{Tikhonov Regularized Adversarial Estimator.} 
As a first step, we consider a Tikhonov regularized adversarial estimator. The regularized population criterion and the corresponding regularized population solution is defined as:
\begin{align}\label{eqn:populatoin_objective}
    h_{*} := \argmin_{h\in \bar{\Hcal}} \|\Tcal(h_0 - h)\|^2 + \lambda \|h\|^2 
\end{align}
We define the Tikhonov regularized adversarial estimator, as the solution to an empirical analogue the adversarial formulation of the regularized population criterion, replacing also $\bar{\Hcal}$ with $\Hcal$ (as defined in Assumption~\ref{ass:realizable}):
\begin{align}\label{eqn:tikhonov}
    \hat{h} := \argmin_{h\in \Hcal} \max_{f\in \Fcal} \E_n \Big[ 2\,\Big(m(W;f) - h(X)\, f(Z) \Big) - f(Z)^2 + \lambda h(X)^2 \Big]
\end{align}
An essential implication of the source condition is the ability to control the regularization bias caused by the Tikhonov regularization term, denoted as $h_* - h_0$, as a function of both $\beta$ and the regularization strength $\lambda$ as follows.

\begin{lemma}[Tikhonov Regularization Bias]\label{lem:bias-tikhonov} 
Suppose that the minimum $L_2$-norm solution $h_0$ satisfies the $\beta$-source condition. Then, we have 
\begin{align}
    \|h_* - h_0\|^2 \leq~& \|w_0\|\, \lambda^{\min\{\beta, 2\}} & 
    \|\Tcal(h_* - h_0)\|^2 \leq~& \|w_0\|\, \lambda^{\min\{\beta + 1, 2\}} & 
\end{align}
\end{lemma}
The first inequality is widely recognized (c.f. \citep[Theorem 1.4.]{cavalier2011inverse}). In our analysis we also incorporate the second inequality, which bounds the bias in terms of the weak metric.

Having established control over the bias, we are now poised to formally demonstrate the previously mentioned estimation rates for the Tikhonov regularized estimator presented in \eqref{eqn:tikhonov}. In the following theorem, we make use of the concept of critical radius, a standard measure for quantifying convergence in nonparametric regression \citep{wainwright2019high}. In particular, when the function classes under consideration are VC classes, the critical radius is determined to be $\delta_n = O(n^{-1/2})$. It is worth noting that the critical radius can be readily computed for various function classes such as Sobolev balls, Holder balls, and sparse linear models (for detailed calculations see \citep{dikkala2020minimax,wainwright2019high,foster2019orthogonal}).

\begin{theorem}\label{thm:adv-l2}
    Consider the function spaces:
    \begin{align}
        \Hcal\cdot \Fcal :=~& \{(x, z) \to h(x)\, f(z): h\in \Hcal, f\in \Fcal\} &
        m \circ \Fcal :=~& \{w \to m(w; f): f\in \Fcal\}
    \end{align}
    Suppose that $m(W; f), h(X), f(X)$ are a.s. absolutely bounded, uniformly over $h\in \Hcal, f\in \Fcal$. Let $\delta_n=\Omega\left(\sqrt{({\log\log(n) + \log(1/\zeta)})/{n}}\right)$ be an upper bound on the critical radius of $\shull(\Hcal\cdot \Fcal)$, $\shull(m\circ \Fcal)$, $\shull(\Fcal)$ and $\shull(\Hcal)$.\footnote{The star hull of a function space is defined as $\shull(\Gcal)=\{\gamma g: g\in \Gcal, \gamma \in [0, 1]\}$. Note that by the linearity of the moment $\shull(m\circ \Fcal)=m\circ \shull(\Fcal)$.} 
Assume (a) the minimum $L_2$-norm solution $h_0$ satisfies the $\beta$-source condition,(b) realizability, $h_{*} \in \Hcal$, (c) closedness assumption, which says that the function space $\Fcal$ satisfies 
    \begin{align}
        \forall h\in \Hcal: \E[h_0(X) - h(X)\mid Z=\cdot] \in \Fcal
    \end{align}
    and (d) the moment $m$ satisfies mean-squared-continuity:
    \begin{align}\label{eqn:msc-l2}
        \forall f\in \Fcal: \E[m(W;f)^2] \leq O\left(\|f\|^2\right). 
    \end{align}
    Then the estimate $\hat{h}$ from Equation~\eqref{eqn:tikhonov}, for any $\lambda < 2$, satisfies that w.p. $1-\zeta$:
    \begin{align}
    \|\hat{h}-h_0\|^2 =~& O\left(\frac{\delta_n^2}{\lambda} + \|w_0\|\, \lambda^{\min\{\beta, 1\}} \right) &
    \|\Tcal(\hat{h}-h_0)\|^2 =~& O\left(\delta_n^2 + \|w_0\|\, \lambda^{\min\{\beta + 1, 2\}}\right)
    \end{align}
\end{theorem}
Hereafter, we present the implications of the aforementioned theorem. Firstly, assumptions (b) and (c) are standard, as utilized in \citep{dikkala2020minimax,LiaoLuofeng2020PENE}. In cases where these assumptions are violated, additional calculations allow us to obtain results with extra bias terms, which quantify the extent of the violation. Secondly, if we optimally choose $\lambda\sim \delta_n^{\frac{2}{1 + \min(\beta,1)}}$, so as to minimize the strong metric upper bound in Theorem~\ref{thm:adv-l2}, then we simultaneously get the rates:
\begin{align}
    \|\hat{h}-h_0\|^2 =~& O\left(\delta_n^{2\frac{\min(\beta,1)}{1+\min(\beta,1)}}\right) & 
    \|\Tcal(\hat{h}-h_0)\|^2 =~& O\left(\delta_n^2\right)
\end{align}
For special cases of the setting we cover here, these rates are equivalent to the current state-of-the-art rates, in terms of the weak metric as presented in \citep{dikkala2020minimax}, and the strong metric as shown in \citep{bennett2023minimax}. Moreover, unlike these prior works, our analysis offers simultaneously state-of-the-art guarantees for both metrics. Consequently, our analysis shows that for the one-step Tikhonov regularized adversarial estimator there is no trade-off in terms of which metric to prioritize (up to multiplicative constants). We now provide a more detailed comparison of our results with existing ones.

\vspace{.5em}\noindent\textbf{Comparison with \citep{LiaoLuofeng2020PENE}.} The authors proposed the same estimator and focus on $\Hcal$ and $\Fcal$ that are neural networks classes.
By naively following their analysis and extending it to cases involving general function classes, we obtain the following result:
\begin{align}
    \|\hat h-h_0 \|^2 =  O(\delta_n/\lambda + \lambda^{\min(\beta,2))}).
\end{align}
Consequently, for $\beta\leq 2$, the resulting rate is of $\|\hat h-h_0 \|$ is $\delta_n^{\beta/2(1+\beta)}$, while for $\beta\geq 2$, the corresponding rate is $\delta_n^{1/3}$. 
As depicted in Figure~\ref{fig:comparison}, this rate is slower compared to the one we obtained. The looser analysis in their work stems from their failure to fully exploit the strong convexity of the loss function \eqref{eqn:populatoin_objective} with respect to $h$ in terms of the strong metric. In our analysis, we refine the bound by employing the localization technique, which entails bounding the empirical process term using the errors $\|\hat h-h_0\|$ and $\|\Tcal(\hat h-h_0)\|$.

\vspace{.5em}\noindent\textbf{Comparison with \citep{bennett2023minimax}.} They propose an estimator that achieves $\| \hat h-h_0 \|^2=O(\delta_n)$ when $\beta\geq 1$.
Although their estimator can relax assumption (c) by replacing it with a realizability assumption of the form ``there exists $w_0 \in \Hcal$ such that $h_0 = \Tcal^* w_0$'', it remains unclear whether their estimator can achieve $\|\Tcal(\hat h-h_0)\|=O(\delta_n^2)$. Furthermore, their estimator does not provide any guarantee when $\beta < 1$.

\vspace{.5em}\noindent\textbf{Comparison with \citep{dikkala2020minimax}.}  They derive an estimator that achieves $\|\Tcal(\hat h-h_0)\|^2=O(\delta_n^2)$. While their result does not necessitate the source condition, their estimator lacks any guarantee regarding the convergence in terms of the strong metric $\|\hat h-h_0\|$, especially when the primal inverse problem does not have a unique solution.

\section{Iterated Tikhonov Regularized Estimator}

One limitation of the result in the previous section is its lack of adaptability to the degree of ill-posedness in the inverse problem, particularly for larger values of $\beta$ corresponding to milder problems. Ideally, as $\beta\to\infty$, it is expected that the strong metric rates would converge to $O(\delta_n^2)$. Indeed, when the operator $\Tcal$ coincides with the identity operator, leading to $\beta=\infty$, nonparametric IV regression reduces to standard nonparametric regression. In such cases, achieving an $O(\delta_n^2)$ rate is well-known \citep{wainwright2019high}. The observed convergence rate saturation after some small value of $\beta$ is a recognized drawback of Tikhonov regularization \citep{Carrasco2007}.
To address this limitation, we propose an iterated Tikhonov regularized adversarial estimator. At each iteration $t$ (commencing with $h_{*,0}=\hat{h}_0=0$) \footnote{Technically, we can set any function as $\hat{h}_0$.},
we consider the population criterion: 
\begin{align}\label{eqn:iter-tikh-reg-solution}
    h_{*,t} := \argmin_{h\in \bar{\Hcal}} \|\Tcal(h_0 - h)\|^2 + \lambda \|h-h_{*,t-1}\|^2
\end{align}
and the corresponding empirical criterion:
\begin{align}\label{eqn:iter-tikhonov}
    \hat{h}_t = \argmin_{h\in \Hcal} \max_{f\in \Fcal} \E_n \Big[ 2\,\Big(m(W;f) - h(X)\, f(Z) \Big) - f(Z)^2 + \lambda \Big( h(X) - \hat{h}_{t-1}(X) \Big)^2 \Big]
\end{align}
We will show that by choosing $t$ and $\lambda$ appropriately then with high probability, as long as $n$ 
is larger than some constant then this iterated estimator achieves a strong metric rate of $\approx \delta_n^{2 \frac{\beta}{\beta+2}}$ that adapts to large values of $\beta$, \ie, becomes faster as $\beta\to \infty$ and converges to $\delta_n^2$ in the limit.

The limitation observed in the previous analysis stems from the inability of the bias in \pref{lem:bias-tikhonov} to adapt to higher values of $\beta$. However, the iterative nature of Tikhonov regularization possesses a crucial characteristic that mitigates this issue. It is demonstrated that the regularization bias of the $t$-th population iterate, denoted as $h_{*,t} - h_0$, is significantly smaller than that of the non-iterated one, particularly when $\beta$ is large.

\begin{lemma}[Iterated Tikhonov Regularization Bias]\label{lem:bias-iter-tikhonov} If $h_0$ is the minimum $L_2$-norm solution to the linear inverse problem, and satisfies the $\beta$-source condition, then the solution to the 
$t$-th iterate of Tikhonov regularization $h_{*,t}$, defined in Equation~\eqref{eqn:iter-tikh-reg-solution}, 
with $h_{*,0}=0$, satisfies that:
\begin{align}
    \|h_{*,t} - h_0\|^2 \leq~& \|w_0\|^2  \lambda^{\min\{\beta, 2t\}} & 
    \|\Tcal(h_* - h_0)\|^2 \leq~& \|w_0\|^2 \lambda^{\min\{\beta + 1, 2t\}} & 
\end{align}
\end{lemma}

With such an improved bias control at hand we can prove our claimed guarantee for the iterated tikhonov regularizedd estimator. The proof is based on starkly different approach to controlling the variance part, than what we used in Theorem~\ref{thm:adv-l2}. This alternative approach is required, so as to avoid compounding of bias terms over the $t$ iterates.

\begin{theorem}[Iterated Tikhonov]\label{thm:adv-l2-iter}
    Under Assumption~\ref{ass:ubound} and further assuming that 
    $\delta_n=\Omega\left(\sqrt{\frac{\log\log(n) + \log(t/\zeta)}{n}}\right)$ 
    upper bounds the critical radii of the function classes define in Theorem~\ref{thm:adv-l2} and that $\delta_n$ also upper bounds the critical radius of $\shull(\Hcal-\Hcal):=\{\gamma\,(h-h'):h,h'\in \Hcal, \gamma\in [0,1]\}$.
    Assume the conditions (a), (c), (d) of Theorem~\ref{thm:adv-l2}, and (b') the realizability $h_{*,\tau} \in \Hcal$ for any $\tau\leq t$. Let $M_{\leq t} = \max\left\{\lambda, \max_{k\in [t]} \|\Tcal(h_{*,k}-h_0)\|_{L_{\infty}}^2\right\}$. Then the estimate $\hat{h}_t$ from Equation~\eqref{eqn:iter-tikhonov}, for and $t\geq 1$ and $\lambda\leq 1$ satisfies w.p. $1-4\zeta$:
    \begin{align}
    \|\hat{h}_t-h_0\|^2 =~& O\left(16^t \frac{\delta_n^2 M_{\leq t} }{\lambda^2} + \|w_0\|_2 \lambda^{\min\{\beta, 2t\}}\right),\\
    \|\Tcal(\hat{h}_t - h_0)\|^2 =~& O\left(\delta_n^2 + \min\left\{\lambda, \frac{16^t\delta_n^2M_{\leq t-1}}{\lambda}\right\} +  \|w_0\|\, \lambda^{\min\{\beta+1,2t\}}\right).
    \end{align}
\end{theorem}


If one uses the crude bound of $\max_{k\in [t-1]} \|\Tcal(h_{*,k}-h_0)\|_{L_{\infty}}^2=O(1)$, then this theorem yields strong and weak metric rates of:
\begin{align}
 \|\hat{h}_t-h_0\|^2 =~& O\left(16^t \frac{\delta_n^2}{\lambda^2} + \lambda^{\min\{\beta, 2t\}}\right) \\
  \|\Tcal(\hat{h}_t - h_0)\|^2 =~& O\left(\delta_n^2 + \min\left\{\lambda, \frac{16^t\delta_n^2}{\lambda}\right\} +  \lambda^{\min\{\beta+1,2t\}}\right)
\end{align}
If we let $\beta_t = \min\{\beta, 2t\}$ and choose a regularization strength of $\lambda\sim \delta_n^{\frac{1}{\beta_t+2}}$, then we get rates:
\begin{align}
    \|\hat{h}_t-h_0\|^2 =~& O\left(16^t \delta_n^{2\frac{\beta_t}{\beta_t+2}}\right)
\end{align}
If we choose $t=\lceil \min\{\beta/2, \log\log(1/\delta_n)\}\rceil$ then we get a rate of:
\begin{align}
    \|\hat{h}_t-h_0\|^2 =~& O\left(\min\{16^{\beta}, \log(1/\delta_n)\} \delta_n^{2\frac{\beta_t}{\beta_t+2}}\right) = \tilde{O}\left(\delta_n^{2\frac{\beta_t}{\beta_t+2}}\right)
\end{align}
Hence for any constant $\beta$, as $n$ grows, eventually $\log\log(1/\delta_n)\geq \beta$ and we get the rate of $O\left(\delta_n^{2 \frac{\beta}{\beta+2}}\right)$. This rate can also be achieved, even if $\beta$ is allowed to grow with $n$ in the asymptotics, as long as it grows slower than $\log\log(1/\delta_n)$. If $\delta_n \sim n^{-\alpha}$ for some $\alpha>0$, then we note that if we take $t= \lceil \min\{\beta/2, \sqrt{\log(1/\delta_n)}\}\rceil$, then $16^{\sqrt{\log(1/\delta_n)}}=o(n^{\epsilon})$ for any $\epsilon>0$, thus we again conclude the rate $\tilde{O}\left(\delta_n^{2 \frac{\beta_t}{\beta_t+2}}\right)$. This result allows us to claim a rate of $O\left(\delta_n^{2 \frac{\beta}{\beta+2}}\right)$ even if $\beta$ is growing with $n$ in the asymptotics as long as it grows slower than $\sqrt{\log(1/\delta_n)}$, which is considerably faster then $\log\log(1/\delta_n)$.

In summary, as depicted in Figure~\ref{fig:comparison}, when $\beta \geq 2$ and the sample size $n$ is sufficiently large, the iterated algorithm exhibits an improved rate in terms of the strong metric compared to the non-iterated version and the existing work \citep{LiaoLuofeng2020PENE}, where the rate saturates after $\beta=2$.
Notably, as $\beta$ tends to infinity, the rate of the iterated algorithm approaches the fast rate $O(\delta_n^2)$. However, such hyperparameter settings tuned for the strong metric do not yield a fast rate in terms of the weak metric, but rather a rate of $O\left(\delta_n^{2\frac{1+\beta}{2+\beta}}\right)$. 
Hence, there is a trade-off in the choice of estimation hyperparameters, dependent on which metric one is targeting.

\begin{remark}[Smooth Function Spaces]\label{rem:smoothness}
If we further assume that the function space $T\circ (\Hcal-h_0)$ satisfies a smoothness assumption that implies a relationship between the $L_{\infty}$ and $L_2$ norms, \ie, for some $\gamma\leq 1$:
\begin{align}\label{eqn:norm-bounds}
    \|\Tcal(h_{*,t} - h_0)\|_{L_{\infty}} \leq O\left(\|\Tcal(h_{*,t} - h_0)\|_{L_{2}}^{\gamma}\right)
\end{align}
then we immediately get by the bias lemma that $\|\Tcal(h_{*,t} - h_0)\|_{L_{\infty}}^2 = O\left(\lambda^{\gamma \min\{1+\beta, 2t\}}\right)=O\left(\lambda^{\gamma}\right)$ 
Thus the strong metric rate becomes:
\begin{align}
    \|\hat{h}_t-h_0\|^2 =~& O\left(16^t\frac{\delta_n^2}{\lambda^{2-\gamma}} + \lambda^{\min\{\beta, 2t\}}\right)
\end{align}
In this case, with the optimal choice of $\lambda$, for any constant $\beta$, with $t=\lceil \beta/2\rceil$ we get a rate:
\begin{align}
    \|\hat{h}_t - h_0\|^2 =O\left(\delta_n^{2\frac{\beta}{2 - \gamma +\beta}}\right).
\end{align}
As $\gamma \to 1$, this recovers the rate of $\delta_n^{2\frac{\beta}{1+\beta}}$ that we have established in the small $\beta\leq 1$ regime and extends it for larger values of $\beta$. Such a relationship is known for instance to hold for Sobolev spaces and more generally for reproducing kernel Hilbert spaces (RKHS) with a polynomial eigendecay. For instance, Lemma 5.1 of \citep{mendelson2010regularization} shows that if the eigenvalues of an RKHS decay at a rate of $1/j^{1/p}$ for some $p\in (0,1)$, and both $h_{*,t},h_0$ lie in the RKHS with a bounded RKHS norm, then we have that Equation~\eqref{eqn:norm-bounds} holds with $\gamma=1-p$. For such function spaces, we then achieve  rate of $O\left(\delta_n^{2\frac{\beta}{1+p+\beta}}\right)$. For the Gaussian kernel, which has an exponential eigendecay, we can take $p$ arbitrarily close to $0$.
\end{remark}

\section{Conditions for Asymptotically Normal Inference  }

As was already noted, the problem that defines $q_0$ is of exactly the same nature as the primary inverse problem that defines $h_0$. The only difference is that the moment is $\tilde{m}$ instead of the moment $m$ that defines $h_0$, the linear operator is the adjoint $\Tcal^*$
of the operator that defines $h_0$ and the roles of $X$ and $Z$ are reversed, \ie, $X$ becomes the ``exogenous'' or ``conditioning'' variable and $Z$ the ``endogenous''. Hence, the Tikhonov Regularized adversarial estimator for $q_0$ becomes:
\begin{align}\label{eqn:tikhonov_dual}
    \hat{q} := \argmin_{q \in \Qcal} \max_{g\in \Gcal} \E_n[2\,(\tilde m(W;g) - q(Z)\, g(X)) - g(X)^2 + \lambda q(Z)^2]
\end{align}
where $\Qcal$ and $\Gcal$ are defined in Assumption~\ref{ass:realizable}. We can similarly define the iterated version. 
Moreover, we can obtain convergence guarantees for $\hat q$ and its iterated version via \pref{thm:adv-l2} and \pref{thm:adv-l2-iter}.
We will assume that both the primal and dual inverse problems satisfy a source condition for some positive $\beta$ (a very benign assumption, \ie, arbitrary weak source condition).
\begin{assumption}
There exist $w_0\in L_2(X)$ and $\tilde w_0 \in L_2(Z)$ such that for $\beta_h,\beta_q>0$:
\begin{align}
h_0 =~& (\Tcal^* \Tcal)^{\beta_h/2} w_0 &
q_0 =~& (\Tcal \Tcal^*)^{\beta_q/2} \tilde w_0.
\end{align}
\end{assumption}
The condition for the dual problem states that $q_0$
should be primarily supported on the lower spectrum of the  operator $T^*$, equivalently, $a_0$ should be primarily supported on the lower spectrum of the operator $T$ (see Appendix~\ref{app:riesz-smooth} for a more detailed discussion).

We present sufficient conditions for guaranteeing asymptotically normal inference on $\hat{\theta}$. We will denote with $\beta=\max\{\beta_h, \beta_q\}$, to be the largest of the two source condition parameters (\ie, the degree of the most well-posed inverse problem). First, we examine the case where we know which of the two inverse problems satisfies the $\beta$-source condition and without loss of generality we will assume that to be dual inverse problem that defines $q_0$, while the primal inverse problem that defines $h_0$ satisfies an arbitrary weak source condition. These roles can be interchanged without difficulty. Subsequently, we consider a more challenging scenario where we do not know which one satisfies the non-vacuous source condition.

\subsection{Guarantees under Knowledge of the Most Well-Posed Inverse Problem}

Suppose that $\beta=\beta_q \geq \beta_h > 0$. The complementary case follows identically, interchanging the roles of $h$ and $q$. Theorem~\ref{thm:adv-l2} implies that by choosing 
$\lambda \sim\delta_n^{\frac{2}{1+\min(\beta_h, 1)}}$, w.p. $1-O(\zeta)$:
\begin{align}
    \|\hat{h}-h_0\|^2 =~& O\left(\delta_n^{2\frac{\min(\beta_h,1)}{1 + \min(\beta_h,1)}}\right)= o_p(1), &
    \|\Tcal(\hat{h}-h_0)\|^2 =~& O\left(\delta_n^2\right).
\end{align}
Thus we guarantee fast rates of $O(\delta_n^2)$
for the weak metric and consistency with respect to the strong metric; as required by the asymptotic normality Corollary~\ref{cor:functionals-endo}. Here $\delta_n$ is an upper bound on the critical radius of the function classes $\shull(\Hcal\cdot \Fcal)$, $\shull(m \circ \Fcal)$, $\shull(\Fcal)$, $\shull(\Hcal - \Hcal)$. Moreover, since $q_0$ satisfies a $\beta_q$-source condition, we can apply the best of Theorem~\ref{thm:adv-l2} or Theorem~\ref{thm:adv-l2-iter} with an optimal choice of regularization hyperparameter and number of iterations, to get:
\begin{align}
    \|\hat{q} - q_0\|^2 =~& O\left(\min\left\{\delta_n^{2 \frac{\min\{\beta, 1\}}{1 + \min\{\beta, 1\}}}, \delta_n^{2 \frac{\beta}{2 + \beta}}\right\}\right) 
\end{align}
where $\delta_n$ is an upper bound on the critical radius of the function classes $\shull(\Qcal\cdot \Gcal)$, $\shull(\tilde{m} \circ \Gcal)$, $\shull(\Gcal)$, $\shull(\Qcal - \Qcal)$.

Combining the two aforementioned observations, we can construct estimators $\hat{h}, \hat{q}$ for both $h_0$ and $q_0$ to be used in the context of Corollary~\ref{cor:functionals-endo}, that ensure asymptotic normality if the following rate condition is satisfied:
\begin{align}\label{eq:final}
   {n}\, \|\hat{q}-q_0\|^2\, \|\Tcal (\hat{h}-h_0)\|^2 = n\, O\left(\min\left\{\delta_n^{2 \left(1 + \frac{\min\{\beta, 1\}}{1 + \min\{\beta, 1\}}\right)}, \delta_n^{2 \left(1 + \frac{\beta}{2 + \beta}\right)}\right\}\right) =~& o(1)
\end{align}
We can conclude an analogous statement by flipping the role of $h_0$ and $q_0$. The size requirement of the function classes varies depending on the source condition. More specifically, we need 
\begin{align}
    \delta_n = o(n^{-\alpha}),\quad \alpha := \min\left( \frac{1 + \min(\beta,1) }{ 2 + 4 \min(\beta,1) }, \frac{2+\beta}{4 + 4 \beta} \right). 
\end{align}
This rate is illustrated in the blue line in \cref{fig:required_rate}. For $1\leq \beta\leq 2$, it is sufficient to have $\delta_n=n^{-1/3}$, which still permits non-parametric function classes. In the extreme case where $\beta \to 0$, the condition converges to $\delta_n=n^{-1/2}$,
which corresponds to a  parametric rate (\ie, function classes are VC classes). Therefore, our theorem accommodates scenarios where both $h_0$ and $q_0$ exhibit severe ill-posedness. Conversely, in the limit as $\beta\to \infty$, the condition requires $\delta_n=n^{-1/4}$, which is the requirement typically encountered when dealing with functionals of functions defined through regression problems rather than inverse problems (c.f. \citep{chernozhukov2017double}).

\vspace{.5em}\noindent\textbf{Comparison with \citep{bennett2022inference}.} 
To provide a comparison, we examine our condition alongside recent related work by \citep{bennett2022inference}. In \citep{bennett2022inference}, in their study, under the assumption of $\beta=1$, they demonstrate that asymptotically normal inference is feasible as long as
\begin{align}
  {n}O(\delta^3_n) = o(1). 
\end{align}
The corresponding required rate for $\delta_n$ is illustrated in the blue line in \cref{fig:required_rate}. This aligns with our conclusion when instantiating $\beta=1$ in \eqref{eq:final}. However, our condition is significantly more general. Firstly, when $\beta<1$, our condition in \eqref{eq:final} still allows for the possibility of asymptotically normal inference. Specifically, as $\beta \to 0$,
our requirement converges to $\delta_n = o(n^{-1/2})$.
In contrast, \citep{bennett2022inference} do not provide any guarantees when $\beta<1$.
Secondly, when $\beta>1$, we can leverage the potentially larger size of the function classes as $\beta$ increases. Specifically, as $\beta \to \infty$, our requirement becomes $\delta_n^4 = o(n^{-1/4})$. This adaptability to $\beta$ is not obtained in \citep{bennett2022inference}.

\subsection{Source Condition Double Robustness}

The current results do not guarantee double robustness in terms of source conditions. We have presented an estimator that enables asymptotically normal inference when $h_0$ satisfies the $\beta$-source condition and $q_0$ satisfies an arbitrary weak source condition. Similarly, we have established a corresponding statement in a scenario when we interchange their role. In this section, we aim to develop an estimator that achieves double robustness, allowing for asymptotically normal inference in \emph{both} scenarios simultaneously.

To achieve this objective, we propose the following refined methods. For convenience, let 
\begin{align}
    L_n(h) :=~& \max_{f\in \Fcal} \E_n[2(m(W;f)-h(X)\,f(Z)) - f(Z)^2]\\
    \tilde{L}_n(q) :=~& \max_{g \in \Gcal}\E_n[2\,(\tilde m(W;g) - q(Z)\, g(Z)) - g(X)^2 ]
\end{align}
Then we introduce the constrained function classes:
\begin{align}
       \tilde \Hcal :=~& \{h \in \Hcal : L_n(h) - \min_{h\in \Hcal} L_n(h)\leq \mu_{n} \}, \label{eqn:emp-loss-h}
       \tilde \Qcal :=~& \{q \in \Qcal: \tilde{L}_n(q) - \min_{q \in \Qcal}\tilde{L}_n(q)\leq \mu_{n} \}, 
\end{align}
where $\mu_{h,n}$ and $\mu_{q,n}$ are parameters that will be specified later. After defining these version spaces, we follow the same procedure as before, but now using $\tilde \Hcal$ and $\tilde \Qcal$ instead of $\Hcal$ and $\Qcal$, respectively. Here, we construct these version spaces such that with high probability: (1) all the functions in $\tilde \Hcal$ are sufficiently close to $h_0$ in terms of weak metric, even in the absence of the source condition, and (2) $h_{*} \in \tilde \Hcal$ under the $\beta$-source condition (similarly for $\tilde \Qcal$). The first property plays a crucial role in achieving source double robustness, while the second property allows us to apply \cref{thm:adv-l2} to $\tilde \Hcal$ instead of $\Hcal$ and ensure a sufficiently fast rate in terms of strong metric under the $\beta$-source condition.

Formally, the resulting estimators with version spaces possess the following properties. An analogous conclusion can be also derived for $\hat q$. We first consider the non-iterated version. 

\begin{corollary}[Non-iterated version]\label{thm:adv-l3}
Suppose conditions (c) and (d) in \cref{thm:adv-l2}. Then, assuming $h_0 \in \Hcal$ and setting $\mu_n=\Omega(\max(\delta^2_n,\lambda^{\min(\beta+1,2)}))$, with probability $1-\zeta$, we can ensure 
    \begin{align}
        \|\Tcal(\hat{h} - h_0)\|^2 =~& O\left(\mu_n\right).
    \end{align}
If two conditions: (a) a function $h_0$ satisfies $\beta$-source condition, (b) realizability $h_* \in \Hcal$ in \cref{thm:adv-l2} additionally hold, with probability $1-\zeta$, we can ensure 
\begin{align}
     \|\hat{h} - h_0 \|^2 =  O( \delta^2_n/\lambda + \lambda^{\min(\beta,1)}). 
\end{align}
\end{corollary}

Importantly, the first statement concerning the weak metric does not rely on the source condition. Utilizing this theorem, let us consider a scenario where either $h_0$ or $q_0$ satisfies the $\beta$-source condition, but we are uncertain about which one. In such cases, by setting $\mu_n = c (\delta^2_n + \lambda^{\min(\beta+1,2)})$, asymptotically normal inference remains feasible under the condition that
\begin{align}
    n \times \underbrace{(\delta^2_n + \lambda^{\min(\beta+1,2)})}_{\text{weak metric }}\times \underbrace{(\delta^2_n/\lambda + \lambda^{\min(\beta,1)}) }_{\text{strong metric}}=o(1).
\end{align}
By optimally setting $\lambda\sim \delta_n^{\frac{2}{1+\min(\beta, 1)}}$, the above becomes
\begin{align}
    n  \times \underbrace{ \delta^2_n}_{\text{Week metric  }}  \times \underbrace{\delta^\frac{2 \min(\beta,1) }{1 + \min(\beta,1)}_n}_{\text{Strong metric } }= o(1). 
\end{align}
Moreover, since $\beta>0$, the above setting of $\lambda$ also ensures that consistency with respect to the strong metric at some arbitrarily slow rate continues to hold, even when the $\beta$-source condition does not hold for the function but only an $\epsilon$ source condition holds for an arbitrarily small $\epsilon>0$ (see also discussion in the previous section).
This result is noteworthy since, in the previous section using Theorem \ref{thm:adv-l2}, we required prior knowledge regarding which one between $h_0$ and $q_0$ satisfies the source condition.


Next, we consider the iterated version to leverage a higher value of $\beta\geq 1$. At each iteration, we optimize over the constrained spaces ${\tilde \Hcal_t, \tilde \Qcal_t}$, defined similar to $\tilde\Hcal,\tilde\Qcal$, but with an iterate dependent upper bound $\mu_{n,t}$ instead of $\mu_n$.
By replacing $\Hcal$ and $\Qcal$ with $\tilde \Hcal_t$ and $\tilde \Qcal_t$ respectively at each iteration, we obtain estimators $\hat h_t$ and $\hat q_t$. Here is the property of the estimators. 

\begin{corollary}[Iterated version]\label{thm:adv-l3-iter}
Suppose conditions (c), (d) in \cref{thm:adv-l2}. Then, assuming $h_0 \in \Hcal$ and setting $\mu_{n,t} = \Omega(\max(\delta^2_n,\lambda^{\min(\beta+1,2t)}))$, for any $\lambda\leq 1$,  with probability $1-O(\zeta)$, we can ensure 
    \begin{align}
        \|\Tcal(\hat{h}_t - h_0)\|^2 =~& O\left(\mu_{n,t} \right).
    \end{align}
If we further assume that (a) the function $h_0$ satisfies $\beta$-source condition, (b') the realizability assumption that $\forall t \in [\beta/2], h_{*,t} \in \Hcal$ in \cref{thm:adv-l2}, (e) the smoothness assumption that for all $h\in \Hcal$, we have  $\|T(h-h_0)\|_{L_\infty}=O\left(\|T(h-h_0)\|^{\gamma}\right)$, then with probability $1-O(\zeta)$:
\begin{align}
     \|\hat{h}_t - h_0 \|^2 =  O(\delta^2_n/\lambda^{2-\gamma}  + \lambda^{\min(\beta,2t)}). 
\end{align}
\end{corollary}

Again, in contrast to \cref{thm:adv-l2-iter}, the first statement does not impose any source condition. Using this theorem, let us consider a scenario where either $h_0$ or $q_0$ satisfies the $\beta$-source condition, but we are uncertain about which one. In this case, asymptotically normal inference is possible as long as
\begin{align}
    n \times \underbrace{(\delta^2_n + \lambda^{\min(\beta+1,2t)})}_{\text{weak metric}}\times \underbrace{(\delta^2_n/\lambda^{2-\gamma} + \lambda^{\min(\beta,2t)}) }_{\text{strong metric}}=o(1).
\end{align}
By optimizing $\lambda\sim \delta_n^{\frac{2}{\beta+1}}$, we can conclude that when $(\beta+1)/2 \leq t\leq c'
\log\log(1/\delta_n)$, the asymptotically normal inference for $\hat{\theta}$ is possible as long as 
\begin{align}
    n \times  \underbrace{\delta^{2}_n}_{\text{Weak metric}}\times  \underbrace{\delta^{2\frac{\beta-1 + \gamma}{\beta+1} }_n}_{\text{Strong metric}}    =o(1). 
\end{align}
Moreover, for $\beta\geq 1$, this setting of $\lambda$ is $o(\delta_n)$ and therefore always ensures also consistency with respect to the strong metric, at some arbitrarily slow rate continues to hold, even when the $\beta$-source condition does not hold for the function but only an $\epsilon$ source condition holds for an arbitrarily small $\epsilon>0$.
Interestingly, this conclusion suggests that in the scenario of source double robustness, our focus should be on achieving a fast rate for the weak metric, even if it comes at the expense of the strong metric. \footnote{Recall, by optimizing $\lambda$ just for the strong metric and for $\gamma=0$, we could achieve $\delta_n^{2\beta/(2+\beta)}$, which is faster than $ \delta^{\frac{2(\beta-1)}{\beta+1} }_n$.} 

To summarize, by combining the best of the two product rate conditions from the two Corollaries and choosing the best of the two conditions, based on the known value of $\beta$, we get the size requirement on the function classes:  
\begin{align}
    \delta_n = o(n^{-\alpha}),\quad \alpha := \min\left( \frac{1 + \min(\beta,1) }{ 2 + 4 \min(\beta,1) }, \frac{1+\beta}{2\gamma + 4 \beta} \right). 
\end{align}
Note that the smoothness assumption always holds when $\gamma=0$, and therefore, in the absence of any smoothness condition the condition is derived from the latter formula with $\gamma=0$.

\begin{remark}[Computational Aspects for RKHS function spaces]
    As an example of how to address the computational aspect of the proposed constrained regularized estimator, we investigate the case of RKHS hypothesis spaces. We will consider the computation of $\hat{h}$, but an analogous approach can be taken for $\hat{q}$. At each iteration of the iterative and constrained Tikhonov regularized estimator, we need to solve a computational problem of the form
    $$\min_{h\in \Hcal: L_n(h) \leq c} L_n(h) + \lambda \E_n[(h(X) - \bar{h}(X))^2]$$
    with $L_n(h)$ as defined in Equation~\eqref{eqn:emp-loss-h} and where $\Hcal$ and $\Fcal$ are norm-constrained RKHS, i.e. $\Hcal$ corresponds to an RKHS with some positive definite kernel $k_{\Hcal}(x,x')$ and with the constrained that $\|h\|_{\Hcal}^2\leq B$, with $\|\cdot\|_{\Hcal}$ the RKHS norm and some constant $B\geq 0$. Similarly $\Fcal$ is an RKHS with kernel $k_{\Fcal}(z,z')$ and norm constraint $\|f\|_{\Fcal}^2 \leq B$, with $\|\cdot\|_{\Fcal}$ the RKHS norm. First consider the internal optimization problem for each $h$, which appears within the definition of $L_n$:
    \begin{align}
       \max_{f\in \Fcal} \E_n[2(m(W;f)-h(X)\,f(Z)) - f(Z)^2]
    \end{align}
    By well-known representer theorems \citep{KIMELDORF197182,Scholkopf2001representer} of solutions of optimization problems over RKHS functions, and similar to Proposition~10 of \citep{chernozhukov2020adversarial}, any solution to the above problem has to be of the form $f=\Psi'\beta$, for some operator $\Psi: \Hcal \times \R^{2n}$, where each row of $\Psi$ corresponds to known linear combinations of  evaluations of the infinite dimensional Mercer feature map, associated with kernel $k_{\Fcal}$, at some point $\tilde{Z}_i$ that is some modification of the sample $Z_i$. Most notable $K=\Psi\Psi'\in \R^{2n\times 2n}$ corresponds to an appropriately defined empirical finite dimensional kernel matrix. With this property, the optimization problem can be written as a finite $2n$-dimensional parametric constrained concave optimization problem, which is easy to solve:
    \begin{align}
        \max_{\gamma\in \R^{2n}: \gamma'K_{\Fcal}\gamma \leq B} 2 \gamma'c_1 - \gamma' M \gamma
    \end{align}
    for some appropriately defined fixed vector $c_1$ and appropriately defined $2n\times 2n$ symmetric positive semi-definite matrices $K_{\Fcal}$ and $M$. These fixed finite-dimensional vectors and matrices can be calculated with a polynomial in $n$ evaluations of the kernel and the function $h$. The finite-dimensional maximization problem can be solved by ``lagrangifying'' the constraint, computing a closed form solution for each candidate Lagrangian coefficient and performing a line search over the scalar Lagrangian coefficient (see e.g. Appendix~E.3 of \citep{dikkala2020minimax}). Given this solution $\gamma$ to the maximization problem, 
we can also apply the representer theorems for the minimization problem over $h$ (see e.g. Proposition~12 of \citep{chernozhukov2020adversarial}). The minimizer will have a finite $n$-dimensional representation $\Phi\beta$, where each row of $\Phi$ corresponds to the evaluation of the infinite dimensional Mercer feature map, associated with kernel $k_{\Hcal}$, at a sample point $X_i$. The loss function $L_n$ can then be written in the form $\gamma'c_2 - \beta'c_3\, \gamma'c_4$, where $c_2, c_3, c_4$ are finite dimensional fixed vectors that can be calculated with polynomial evaluations of the kernels $k_{\Hcal}$ and $k_{\Fcal}$. The minimization problem can then also be written as a finite dimensional parametric optimization of the form:
    \begin{align}
        \min_{\substack{\beta\in \R^{n}:
        \beta'K_{\Hcal}\beta\leq B \text{ and } \gamma'c_2 - \beta'c_3\, \gamma'c_4 \leq c}} - \beta'c_3\, \gamma'c_4  + \beta'c_5 + \beta'Q\beta
    \end{align}
    for some fixed polynomially computable vectors $c_1, c_2, c_3, c_4, c_5$ and symmetric positive semi-definite matrices $K_{\Hcal}, Q$ (these vectors relate to evaluations of the kernel at pairs of sample points $X_i, X_i'$ and $\tilde{Z}_i, \tilde{Z}_i'$, as well as evaluations of the fixed function $\bar{h}$ used as a regularization center, at sample points $X_i$). This is a constrained convex minimization problem with convex constraints and can be solved in a computationally efficient manner. In fact, it only involves quadratic and linear constraints and objectives and can solved very fast by modern computational toolboxes for quadratic constrained convex optimization via interior point methods. Moreover, since we only have two constraints, they can be ``lagrangified'' and the Lagrange multipliers calculated via grid search over two scalars, turning it into an un-constrained quadratic program for each candidate pair of Lagrange multipliers, which is solvable in closed form.
\end{remark}

\bibliographystyle{abbrvnat}
\bibliography{ref1}
\appendix

\section{Preliminary Lemmas}

We will present here a concentration inequality lemma that is used throughout the proofs. See \citep{foster2019orthogonal} for a proof of this theorem:
\begin{lemma}[Localized Concentration, \citep{foster2019orthogonal}]\label{lem:concentration}
For any $h\in \Hcal := \times_{i=1}^d \Hcal_i$ be a multi-valued outcome function, that is almost surely absolutely bounded by a constant. Let $\ell(Z; h(X))\in \R$ be a loss function that is $O(1)$-Lipschitz in $h(X)$, with respect to the $\ell_2$ norm. Let $\delta_n=\Omega\left(\sqrt{\frac{d\,\log\log(n) + \log(1/\zeta)}{n}}\right)$ be an upper bound on the critical radius of $\shull(\Hcal_i)$ for $i\in [d]$. Then for any fixed $h_*\in \Hcal$, w.p. $1-\zeta$: 
\begin{align}
    \forall h\in \Hcal: \left|(\E_n - \E)\left[\ell(Z; h(X)) - \ell(Z; h_0(X))\right]\right| = O\left(d\, \delta_n \sum_{i=1}^d \|h_i - h_{i,0}\|_{L_2} + d\,\delta_n^2\right)
\end{align}
If the loss is linear in $h(X)$, i.e. $\ell(Z; h(X) + h'(X)) = \ell(Z; h(X)) + \ell(Z; h'(X))$ and $\ell(Z;\alpha h(X)) = \alpha \ell(Z;h(X))$ for any scalar $\alpha$, then it suffices that we take $\delta_n=\Omega\left(\sqrt{\frac{\log(1/\zeta)}{n}}\right)$ that upper bounds the critical radius of $\shull(\Hcal_i)$ for $i\in [d]$.
\end{lemma}

\section{Proof of Lemma~\ref{lem:bias-tikhonov}}

\begin{proof}
The proof of this lemma has been discussed in \citep[Theorem 1.4.]{cavalier2011inverse}. 
For simplicity, we consider a compact linear operator $\Tcal$. 
The non-compact setting can be analyzed similarly, but with more complex singular value decomposition. 
Recall from \cref{sec: intro} that the compact linear operator $\Tcal$ admits  a countable singular value decomposition $\{\sigma_i, v_i, u_i\}_{i=1}^{\infty}$ with $\sigma_1 \geq \sigma_2 \geq \ldots$. 
Let $h_0$ and $h^*$ be the minimum norm solution $h_0 = \argmin_{h \in \Hcal: \Tcal h = r_0}\|h_0\|^2$ and the regularized target $h^* = \argmin_{h \in \Hcal}\|\Tcal\prns{h - h_0}\|+\lambda\|h\|^2$.  
Then we can write 
$h_0 = \sum_{i=1}^\infty a_{i,0} v_i$ and  $h_*=\sum_{i=1}^\infty a_{i,*} v_i$ for some square summable sequences $\braces{a_{i,0}: i = 1, \dots, }$ and $\braces{a_{i,*}: i = 1, \dots, }$.

First, we note that $h_0$ does not place any weight on eigenfunctions with zero singular value, namely, $a_{i, 0} = 0$ for any $i$ such that $\sigma_i = 0$. Otherwise, $h_0$ wouldn't have the minimum norm. Moreover, since it satisfies the source condition, we have $h_0 = (\Tcal^*\Tcal)^{\beta/2} w_0$ for some $w_0 \in L_2(X)$. It is easy to check that 
\begin{align}
    \|w_0\|^2 = \sum_{i=1}^\infty 1\braces{\sigma_i \ne 0} \frac{\langle h_0 ,v_i\rangle^2}{\sigma_i^{2\beta}} = \sum_{i=1}^\infty 1\braces{\sigma_i \ne 0} \frac{a_{i, 0}^2}{\sigma_i^{2\beta}}.
\end{align}

We now derive the form of the regularized target $h^*$. Note that the Tikhonov regularized objective can be written as follows:
    \begin{align}
        \|\Tcal(h_0 - h)\|^2 + \lambda \|h\|^2 = \sum_{i=1}^\infty \sigma_i^2 (a_{i,0} - a_i)^2 + \lambda a_i^2.
    \end{align}
    Thus the optimal solution can be derived by taking the first order condition for each $i$, as:
    \begin{align}
        -\sigma_i^2 (a_{i,0}-a_i) + \lambda a_i = 0 \implies a_{i,*} = \frac{\sigma_i^2}{\sigma_i^2 + \lambda} a_{i,0}.
    \end{align}
    It follows that 
        \begin{align}
        \|h_* - h_0\|^2 =~& \sum_{i=1}^\infty (a_{i,0} - a_{i,*})^2 = \sum_{i=1}^\infty a_{i,0}^2 \frac{\lambda^2}{(\sigma_i^2 +\lambda)^2}\\
        =~& \sum_{i=1}^\infty 1\{\sigma_i\neq 0\} a_{i,0}^2 \frac{\lambda^2}{(\sigma_i^2 +\lambda)^2}\\
        =~& \sum_{i=1}^\infty 1\{\sigma_i\neq 0\} \frac{a_{i,0}^2}{\sigma_i^{2\beta}} \frac{\sigma_i^{2\beta}\lambda^2}{(\sigma_i^2 +\lambda)^2} \leq \|w_0\|^2_2 \max_{i} \frac{\sigma_i^{2\beta}\lambda^2}{(\sigma_i^2 +\lambda)^2},
    \end{align}
    Now we show the desired conclusion by  showing   \begin{align}
        \frac{\sigma_i^{2\beta}\lambda^2}{(\sigma_i^2 +\lambda)^2} =  \lambda^{\min\{\beta, 2\}}. 
    \end{align} 
    When $\beta \ge 2$, we have 
    \begin{align}
        \frac{\sigma_i^{2\beta}\lambda^2}{(\sigma_i^2 +\lambda)^2} \le \lambda^2 \max_{i}\sigma_i^{2(\beta-2)} = \lambda^2.
    \end{align}
    The last inequality holds because by Jensen's inequality we can easily prove $\|\Tcal\|\leq 1$, so that  the maximum singular value of the operator $\Tcal$ is no larger than $1$. 

    When $\beta < 2$, we consider the function
    \begin{align}
        f(x) = \frac{x^\beta \lambda^2}{(x+\lambda)^2}.
    \end{align}
    By the first order optimality condition, this function  
    is maximized at $x = \lambda \beta (2-\beta)^{-1}$ and its maximum is 
    \begin{align}
        \frac{x^{\beta} \lambda^2}{(x+\lambda)^2} \leq  \lambda^{\beta} \frac{\beta^{\beta}(2-\beta)^{2-\beta}}{4} \leq \lambda^{\beta}. 
    \end{align}
    
    For the second part of the lemma, note that:
    \begin{align}
        \|\Tcal(h_*-h_0)\|^2 = \sum_{i=1}^{\infty} \sigma_i^2 (a_{i,0}-a_{i,*})^2 = \sum_{i=1}^\infty 1\braces{\sigma_i \ne 0} \frac{a_{i,0}^2}{\sigma_i^{2\beta}} \frac{\sigma_i^{2(\beta+1)}\lambda^2}{(\sigma_i^2 +\lambda)^2} \leq \|w_0\|^2 \max_{i} \frac{\sigma_i^{2(\beta+1)}\lambda^2}{(\sigma_i^2 +\lambda)^2}
    \end{align}
    Then we can prove the desired conclusion by following the steps above with $\beta$ replaced by $\beta + 1$. 

\qed\end{proof}

\section{Proof of Theorem~\ref{thm:adv-l2}}

\begin{proof}
    Anytime we denote a norm $\|\cdot\|$ in this proof, we mean $\|\cdot\|_{L_2}$. All constants $c_i$ in the proof, denote appropriately large universal constants. 
    
    \vspace{.5em}\noindent\textbf{Strong convexity with $L_2$ metric.} The first key realization is that the regularized criterion is strongly convex not just with respect to the weak metric, but with respect to the strong $L_2$ metric, due to the Tikhonov regularizer. More formally, consider the loss function:
    \begin{align}
        L(\tau) := \|\Tcal(h_0 - h -\tau \nu)\|^2 + \lambda \|h + \tau\nu\|^2.
    \end{align}
    Note that this loss function is quadratic in $\tau$ and strongly convex with:
    \begin{align}
        \partial_{\tau}^2 L(\tau) = \|\Tcal\nu\|^2 + \lambda \|\nu\|^2.
    \end{align}
    Since $h_*$ optimizes the regularized criterion, we have that if we take $h=h_*$ in the definition of $L(\tau)$ and if we take $\nu=\hat{h} - h_*$, then the function $L(\tau)$ is minimized at $\tau=0$. Moreover, since the space $\Hcal$ is convex, any function $h_* + \tau(\hat{h}-h_*)\in \Hcal$ for all $\tau\in [0,1]$ and therefore the derivative with respect to $\tau$ must be non-negative at $\tau=0$ (as otherwise an infinitesimal step in the direction $\nu$ would have decreased the regularized criterion; contradicting the optimality of $h_*$). Thus by an exact second order Taylor expansion, we have that:
    \begin{align}
        L(1) - L(0) = \partial_{\tau}L(0) + \partial_{\tau}^2 L(0) \geq \partial_{\tau}^2 L(0)
    \end{align}
    Instantiating the definition of $L$ and re-arranging, we then get:
    \begin{align}
        \lambda \|\hat{h}-h_*\|^2 +  \|\Tcal(\hat{h} - h_*)\|^2\leq \|\Tcal(h_0 - \hat{h})\|^2 - \|\Tcal(h_0 - h_*)\|^2  + \lambda \left(\|\hat{h}\|^2 - \|h_*\|^2\right) 
    \end{align}
    
    \vspace{.5em}\noindent\textbf{Convergence of weak metric excess risk.} Now we need to argue that $\|\Tcal(h_0 - \hat{h})\|$ is small. 
    \begin{lemma}[Weak Metric with Regularizer]\label{lem:weak-conv-tikh}
    If $\hat{h}$ optimizes the regularized objective:
    \begin{align}
        \hat{h} = \argmin_{h\in \Hcal} \max_{f\in \Fcal} \E_n[2(m(W;f)-h(X)\,f(Z)) - f(Z)^2] + R_n(h)
    \end{align}
    and the assumptions of Theorem~\ref{thm:adv-l2} are satisfied, then for any $h_*\in \Hcal$:
    \begin{align}
        \|\Tcal(h_0 - \hat{h})\|^2 - \|\Tcal(h_0 - h_*)\|^2 \leq \underbrace{2\|\Tcal(h_0 - h_*)\|^2  + R_n(h_*) - R_n(\hat{h})}_{\text{``Bias''}} + \underbrace{O\left(\delta_n \|\Tcal(\hat{h}-h_*)\| + \delta_n^2\right)}_{\text{``Variance''}}
    \end{align}
    \end{lemma}
    \begin{proof}
    Letting $f_{h}=\Tcal (h_0 -h)$
    and $\Delta(h)=\|T(h_0-h)\|^2$. Noting that $f_h\in \Fcal$ and invoking Lemma~\ref{lem:concentration}:
    \begin{align}
        \Delta(\hat{h}) :=~& \|\Tcal(h_0 - \hat{h})\|^2\\
        =~& \E[2\,(m(W; f_{\hat{h}}) - \hat{h}(X)\, f_{\hat{h}}(Z)) - f_{\hat{h}}(Z)^2] 
\tag{Equation~\eqref{eqn:equivalence-adv-weak}}
        \\
        \leq~& \E_n[2\,(m(W; f_{\hat{h}}) - \hat{h}(X)\, f_{\hat{h}}(Z)) - f_{\hat{h}}(Z)^2] \\
        ~&~~ + O\left(\delta_n \left(\sqrt{\E[m(W;f_{\hat{h}})^2]} + \sqrt{\E[\hat{h}(X)^2 f_{\hat{h}}(Z)^2]}+ \|f_{\hat{h}}\|\right) + \delta_n^2\right)\tag{Lemma~\ref{lem:concentration}} \\
        \leq~& \E_n[2\,(m(W; f_{\hat{h}}) - \hat{h}(X)\, f_{\hat{h}}(Z)) - f_{\hat{h}}(Z)^2] + O\left(\delta_n \|f_{\hat{h}}\| + \delta_n^2\right)\tag{Equation~\eqref{eqn:msc-l2}}\\
        \leq~& \sup_{f\in \Fcal} \E_n[2\,(m(W; f) - \hat{h}(X)\, f(Z)) - f(Z)^2] + O\left(\delta_n \|f_{\hat{h}}\| + \delta_n^2\right) \tag{$f_{\hat{h}}\in \Fcal$}\\
        \leq~& \sup_{f\in \Fcal} \E_n[2\,(m(W; f) - h_*(X)\, f(Z)) - f(Z)^2] + O\left(\delta_n \|f_{\hat{h}}\| + \delta_n^2\right) + \underbrace{R_n(h_*)- R_n(\hat{h})}_{I_n}\\
        \leq~& \sup_{f\in \Fcal} \E[2\,(m(W; f) - h_*(X)\, f(Z)) - f(Z)^2] + O\left(\delta_n \|f\| + \delta_n \|f_{\hat{h}}\| + \delta_n^2\right) + I_n\\
        =~& \sup_{f\in \Fcal} \E[2\,(h_0(X) - h_*(X)) f(Z) - f(Z)^2] + O\left(\delta_n \|f\| + \delta_n \|f_{\hat{h}}\| + \delta_n^2\right) + I_n\\
        =~& \sup_{f\in \Fcal} \E[2\,(h_0(X) - h_*(X)) f(Z) - \frac{1}{2} f(Z)^2] + O\left(\delta_n \|f_{\hat{h}}\| + \delta_n^2\right) + I_n \tag{AM-GM}\\
        \leq~& 2 \|\Tcal(h_0 - h_*)\|^2 + O\left(\delta_n \|f_{\hat{h}}\| + \delta_n^2\right) + I_n 
    \tag{Equation~\eqref{eqn:equivalence-adv-weak}}
    \end{align}
    Noting also that $\|f_{\hat{h}}\| = \|\Tcal(h_0 - \hat{h})\|$, we conclude that:
    \begin{align}
        \|\Tcal(h_0 - \hat{h})\|^2 - \|\Tcal(h_0 - h_*)\|^2\leq \|\Tcal(h_0 - h_*)\|^2 + O\left(\delta_n \|\Tcal (h_0-\hat{h})\| + \delta_n^2\right) + R_n(h_*)- R_n(\hat{h})
    \end{align}
    Since $\|\Tcal(h_0-\hat{h})\|\leq \|\Tcal(h_0-h_*)\| + \|\Tcal(\hat{h}-h_*)\|$, we get:
    \begin{align}
        \|\Tcal(h_0 - \hat{h})\|^2 - \|\Tcal(h_0 - h_*)\|^2\leq 2 \|\Tcal(h_0 - h_*)\|^2 + O\left(\delta_n \|\Tcal (\hat{h}-h_*)\| + \delta_n^2\right) + R_n(h_*)- R_n(\hat{h})
    \end{align}
    \qed\end{proof}

    \noindent\textbf{Combining convexity and convergence.} Applying Lemma~\ref{lem:weak-conv-tikh}, for $R_n(h)=\lambda\|h\|_{2,n}^2:=\lambda \E_n[h(X)^2]$, with the strong convexity lower bound and letting $I_n = \lambda\left(\|h_*\|_{2,n}^2 - \|\hat{h}\|_{2,n}^2\right)$ and $I=\lambda \left(\|h_*\|^2 - \|\hat{h}\|^2\right)$ we get:
    \begin{align}
        \lambda \|\hat{h}-h_*\|^2 +  \|\Tcal(\hat{h} - h_*)\|^2\leq~& 2\|\Tcal(h_0 - h_*)\|^2  + O\left(\delta_n \|\Tcal(\hat{h}-h_*)\| + \delta_n^2\right) + I_n - I
    \end{align}
    Applying the AM-GM inequality to the term $\delta_n\|\Tcal(\hat{h}-h_*)\|$, and re-arranging we can derive:
    \begin{align}
        \lambda \|\hat{h}-h_*\|^2 + \frac{1}{2}\|\Tcal(\hat{h}-h_*)\|^2 \leq~& 2\|\Tcal(h_0 - h_*)\|^2   + O\left(\delta_n^2\right) + I_n - I
    \end{align}
    
    Consider the discrepancy between the empirical and the population regularizers:
    \begin{align}
        \|h_*\|_{2,n}^2 - \|\hat{h}\|_{2,n}^2 - (\|h_*\|^2 - \|\hat{h}\|^2) = (\E_n - \E)[h_*(X)^2 - \hat{h}(X)^2]
    \end{align}
    Since this is the difference of two centered empirical processes, we can also upper bound the latter by $O\left(\delta_n \|\hat{h}-h_*\| + \delta_n^2\right)$. Thus we get:
    \begin{align}
        \lambda \|\hat{h}-h_*\|^2 + \frac{1}{2} \|\Tcal(\hat{h} - h_*)\|^2
        \leq~& 2\|\Tcal(h_0 - h_*)\|^2 + O\left(\lambda \delta_n \|\hat{h}-h_*\| + \delta_n^2\right)
    \end{align}
    Applying the AM-GM inequality to the term $\delta_n \|\hat{h}-h_*\|$ and re-arranging:
    \begin{align}
        \frac{\lambda}{2} \|\hat{h}-h_*\|^2 + \frac{1}{2} \|\Tcal(\hat{h} - h_*)\|^2
        \leq~& 2\|\Tcal(h_0 - h_*)\|^2 + O\left(\lambda \delta_n^2 + \delta_n^2\right)\\
        =~& 2\|\Tcal(h_0 - h_*)\|^2 + O\left(\delta_n^2\right)\label{eqn:middle-step-v1}
    \end{align}

    \noindent\textbf{Adding the bias part.} By the triangle inequality:
    \begin{align}
        \|\hat{h}-h_0\|^2 \leq~& 2 \|\hat{h}-h_*\|^2 + 2 \|h_* - h_0\|^2
        \leq~ \frac{4}{\lambda}\|\Tcal(h_0 - h_*)\|^2 + O\left(\frac{\delta_n^2}{\lambda}\right) + 2 \|h_* - h_0\|^2
    \end{align}
    By the Bias Lemma~\ref{lem:bias-tikhonov}, we have:
    \begin{align}
        \|h_* - h_0\|^2 \leq~& O\left(\|w_0\|\, \lambda^{\min\{\beta, 2\}}\right)
        &
        \|\Tcal(h_* - h_0)\|^2 \leq~& O\left(\|w_0\|\, \lambda^{\min\{\beta+1, 2\}}\right)
    \end{align}
    Thus overall we get the desired theorem:
    \begin{align}
        \|\hat{h}-h_0\|^2 = O\left(\frac{\delta_n^2}{\lambda} + \|w_0\|\, \lambda^{\min\{\beta, 1\}} + \|w_0\|\, \lambda^{\min\{\beta,2\}}\right) = O\left(\frac{\delta_n^2}{\lambda} + \|w_0\|\, \lambda^{\min\{\beta, 1\}}\right)
    \end{align}

    \noindent\textbf{Weak metric rate.} Starting from Equation~\eqref{eqn:middle-step-v1}, and adding the weak metric bias, we also get:
    \begin{align}
        \|\Tcal(\hat{h}-h_0)\|^2 \leq~& 2 \|\Tcal(\hat{h}-h_*)\|^2 + 2\|\Tcal(h_*-h_0)\|^2 \leq O\left(\|\Tcal(h_0 - h_*)\|^2 + \delta_n^2\right)\\
        =~& O\left(\delta_n^2 + \|w_0\|\, \lambda^{\min\{\beta+1, 2\}}\right)
    \end{align}
    
\qed\end{proof}

\section{Proof of Lemma~\ref{lem:bias-iter-tikhonov}}

\begin{proof}

The proof of this lemma has been discussed in \citep[Theorem 1.4.]{cavalier2011inverse}. For completeness, we show the case where the operator $\Tcal$ is compact. However, the extension to the non-compact operator is straightforward.
    Let $h_0 = \sum_{i=1}^\infty a_{i,0} v_i$ and $h=\sum_{i=1}^{\infty} a_i v_i$ and $h_*=\sum_{i=1}^\infty a_{i,*} v_i$. By some algebra, we can show that 
    \begin{align}
     a_{i,*} = \frac{(\sigma_i^2+\lambda)^t - \lambda^t}{(\sigma^2_i +\lambda)^t} a_{i,0}. 
    \end{align}
    Then note that since the minimum norm solution $h_0$ to the inverse problem, does not place any weight on eigenfunctions for which $\sigma_i=0$:
    \begin{align}
        \|h_* - h_0\|^2 =~& \sum_{i=1}^\infty (a_{i,0} - a_{i,*})^2 = \sum_{i=1}^\infty a_{i,0}^2 \frac{\lambda^{2t}}{(\sigma_i^2 +\lambda)^{2t} }\\
        =~& \sum_{i=1}^\infty 1\{\sigma_i\neq 0\} a_{i,0}^2 \frac{\lambda^{2t}}{(\sigma_i^2 +\lambda)^{2t} }\\
        =~& \sum_{i=1}^\infty 1\{\sigma_i\neq 0\} \frac{a_{i,0}^2}{\sigma_i^{2\beta}} \frac{\sigma_i^{2\beta}\lambda^{2t}}{(\sigma_i^2 +\lambda)^{2t} } \leq \|w_0\|^2_2 \max_{i} \frac{\sigma_i^{2\beta}\lambda^{2t}}{(\sigma_i^2 +\lambda)^{2t}}
    \end{align}
    where the third equality holds because $a_{i, 0} = 0$ whenever $\sigma_i = 0$ and the last inequality holds because $\sum_{i}1\{\sigma_i\neq 0\} \frac{a_{i,0}^2}{\sigma_i^{2\beta}}=\|w_0\|^2_2$ (see the proof for ).
    Now by a case analysis, we can conclude the lemma by showing that:
    \begin{align}
        \frac{\sigma_i^{2\beta}\lambda^{2t}}{(\sigma_i^2 +\lambda)^{2t}} =  \lambda^{\min\{\beta, 2t\}} 
    \end{align}
    Take $\beta \geq 2t$. Then note that:
    \begin{align}
        \max_{i}  \frac{\sigma_i^{2\beta}\lambda^{2t}}{(\sigma_i^2 +\lambda)^{2t}}  \leq \lambda^{2t} \max_{i} \sigma_i^{2(\beta - 2t)} =  \lambda^{2t} 
    \end{align}
    since the maximum singular value of the operator is less than $1$ recalling $\|\Tcal\|\leq 1$. 
    Take $\beta < 2t$. Then note that the function:
    \begin{align}
        f(x) = \frac{x^\beta \lambda^{2t} }{(x+\lambda)^{2t} }
    \end{align}
    is maximized (by using the first order condition) at:
    \begin{align}
        x = \lambda \beta (2t-\beta)^{-1}
    \end{align}
    and takes value:
    \begin{align}
        \frac{x^{\beta} \lambda^{2t}}{(x+\lambda)^{2t}} \leq = \lambda^{\beta} \frac{\beta^{\beta}(2-\beta)^{2t-\beta}}{2^{2t}} \leq \lambda^{\beta}. 
    \end{align}
    For the second part of the lemma, note that:
    \begin{align}
        \|\Tcal(h_*-h_0)\|^2 = \sum_{i=1}^{\infty} \sigma_i^2 (a_{i,0}-a_{i,*})^2 = \sum_{i=1}^\infty \frac{a_{i,0}^2}{\sigma_i^{2\beta}} \frac{\sigma_i^{2(\beta+1)}\lambda^{2t}}{(\sigma_i^2 +\lambda)^{2t}} \leq \|w_0\|^2_2 \max_{i} \frac{\sigma_i^{2(\beta+1)}\lambda^{2t}}{(\sigma_i^2 +\lambda)^{2t}}
    \end{align}
    Following identical steps but renaming $\beta\to \beta+1$ yields the second part.
\qed\end{proof}

\section{Proof of Theorem~\ref{thm:adv-l2-iter}}

\begin{proof}
    Anytime we denote a norm $\|\cdot\|$ in this proof, we mean $\|\cdot\|_{L_2}$. All constants $c_i$ in the proof, denote appropriately large universal constants. We fix an iteration $t$ and we will denote with $h_*\leftarrow h_{*,t}$, $\hat{h}\leftarrow \hat{h}_t$, $h_{*,-1}\leftarrow h_{*,t-1}$ and $\hat{h}_{-1} \leftarrow \hat{h}_{t-1}$. 
    
    \vspace{.5em}\noindent\textbf{Strong convexity with $L_2$ metric.} The first key realization is that the regularized criterion is strongly convex not just with respect to the weak metric, but with respect to the strong $L_2$ metric, due to the Tikhonov regularizer. More formally, consider the loss function:
    \begin{align}
        L_t(\tau) :=~& \|\Tcal(h_0 - h -\tau \nu)\|^2 + \lambda \|h + \tau\nu-h_{*,-1}\|^2.
    \end{align}
    Note that this loss function is quadratic in $\tau$ and strongly convex with:
    \begin{align}
        \partial_{\tau}^2 L(\tau) = \|\Tcal \nu\|^2 + \lambda \|\nu\|^2.
    \end{align}
    Since $h_{*}$ optimizes the regularized criterion, we have that if we take $h=h_{*}$ in the definition of $L(\tau)$ and if we take $\nu=\hat{h} - h_{*}$, then the function $L(\tau)$ is minimized at $\tau=0$. Moreover, since the space $\Hcal$ is convex, any function $h_{*} + \tau(\hat{h}-h_{*})\in \Hcal$ for all $\tau\in [0,1]$ and therefore the derivative with respect to $\tau$ must be non-negative at $\tau=0$ (as otherwise an infinitesimal step in the direction $\nu$ would have decreased the regularized criterion; contradicting the optimality of $h_{*}$). Thus by an exact second order Taylor expansion, we have that:
    \begin{align}
        L(1) - L(0) = \partial_{\tau}L(0) + \partial_{\tau}^2 L(0) \geq \partial_{\tau}^2 L(0).
    \end{align}
    Instantiating the definition of $L$ and re-arranging, we then get:
    \begin{align}
        \lambda \|\hat{h}-h_{*}\|^2 +  \|\Tcal(\hat{h} - h_{*})\|^2\leq \|\Tcal(h_0 - \hat{h})\|^2 - \|\Tcal(h_0 - h_{*})\|^2  + \lambda \underbrace{\left(\|\hat{h}-h_{*,-1}\|^2 - \|h_{*}-h_{*,-1}\|^2\right)}_{=:-I}.
    \end{align}
    
    \vspace{.5em}\noindent\textbf{Convergence of weak metric excess risk.} Applying Lemma~\ref{lem:weak-conv-tikh} with $R_n(h)=\lambda\|h-\hat{h}_{-1}\|_{2,n}^2$ we get:
    \begin{align}
        \|\Tcal(h_0 - \hat{h})\|^2 - \|\Tcal(h_0 - h_*)\|^2
        \leq O\left(\|\Tcal(h_0 - h_*)\|^2 + \delta_n \|\Tcal (\hat{h}-h_*)\|^2 + \delta_n^2\right)  + \lambda I_n
    \end{align}
    where $I_n:= \|h_*-\hat{h}_{-1}\|_{2,n}^2 - \|\hat{h}-\hat{h}_{-1}\|_{2,n}^2$. Though this bound suffices to get a good rate in one iteration, in multiple iterations the above bound depends on the bias $\|\Tcal(h_0-h_*)\|$. Thus if we inductively invoke such bounds, we will be getting cumulative bias terms appearing in the final bound for the $t$-th iterate. However, since these bias terms are large for small iterates, due to the poorly centered regularization bias, the above excess risk guarantee for $\hat{h}$ is not good enough. We will need to show an excess weak metric risk guarantee for $\hat{h}$ that does not depend on bias.
    \begin{lemma}[Bias-Less Excess Weak Metric with Regularizer]\label{lem:weak-conv-tikh-biasless}
    Suppose that $\hat{h}$ optimizes the regularized objective:
    \begin{align}
        \hat{h} = \argmin_{h\in \Hcal} \max_{f\in \Fcal} \E_n[2(m(W;f) -h(X)\,f(Z)) - f(Z)^2] + R_n(h)
    \end{align}
    Suppose that the conditions of Theorem~\ref{thm:adv-l2-iter} are satisfied {and let $f_{h}:=\Tcal (h_0 - h)$.} Then for any $h_*\in \Hcal$, w.p. $1-2\zeta/t$:
    \begin{align}
        \|\Tcal(h_0 - \hat{h})\|^2 - \|\Tcal(h_0 - h_*)\|^2 \leq \underbrace{R_n(h_*) - R_n(\hat{h})}_{\text{``Regularization Bias''}} + \underbrace{O\left(\delta_n \|\hat{h}-h_*\|\, \|f_{h_*}\|_{L_\infty} + \delta_n \|\Tcal (\hat{h}-h_*)\| + \delta_n^2\right)}_{\mcE_n:=\text{``Variance''}}.
    \end{align}
    \end{lemma}
    \begin{proof}
    Let:
    \begin{align}\label{eqn:game-loss}
        \ell(h, f) = 2(m(W;f) - h(X)\, f(Z)) - f(Z)^2.
    \end{align}
    Note that by Equation~\eqref{eqn:equivalence-adv-weak}:
    \begin{align}\label{eqn:game-loss-equivalence}
        \|\Tcal(h_0 - h)\|^2 =~& \E[2\,(h_0(X) - h(X)) f_{h}(Z) - f_{h}(Z)^2]\\
        =~& \E[2\,(m(W;f) - {h}(X)\, f_{h}(Z)) - f_{h}(Z)^2] = \E[\ell(h, f_h)].
    \end{align}
    Thus we can write:
    \begin{align}
        \|\Tcal(h_0 - \hat{h})\|^2 - \|\Tcal(h_0 - h_*)\|^2=~&  \E[\ell(\hat{h}, f_{\hat{h}}) - \ell(h_*, f_{h_*})].
    \end{align}
    Since the loss function $\ell(h, f)$ is $O(1)$ lipschitz with respect to the functions $(h, f)$, we can apply the localized concentration inequality of Lemma~\ref{lem:concentration} and mean-squared-continuity, to get w.p. $1-\zeta/t$:
    \begin{align}
        \left|(\E_n-\E)[\ell(\hat{h}, f_{\hat{h}}) - \ell(h_*, f_{h_*})]\right| \leq~& O\left(\delta_n \left(\sqrt{\E[m(W; f_{\hat{h}}-f_{h_*})^2] + \E[(\hat{h}(X)\, f_{\hat{h}}(Z)-h_*(X)\,f_{h_*}(Z))^2]}\right)\right)\\
        ~&~~~~ + O\left(\delta_n \|f_{\hat{h}} - f_{h_*}\| + \delta_n^2\right)\\
        \leq~& O\left(\delta_n \sqrt{\E[(\hat{h}(X) - h_*(X))^2 \, f_{h_*}(Z)^2]} + \delta_n \|f_{\hat{h}} - f_{h_*}\| + \delta_n^2\right)\\
        \leq~& O\left(\delta_n \|\hat{h}-h_*\|\, \|f_{h_*}\|_{L_{\infty}} + \delta_n \|f_{\hat{h}} - f_{h_*}\| + \delta_n^2\right) =: \mcE_n
    \end{align}
    Moreover, note that:
    \begin{align}
        \|f_{\hat{h}}-f_{h_*}\| = \|\Tcal(\hat{h}-h_*)\|
    \end{align}
    Thus we get:
    \begin{align}\label{eqn:middle-step}
        \|\Tcal(h_0 - \hat{h})\|^2 - \|\Tcal(h_0 - h_*)\|^2\leq~& \E_n[\ell(\hat{h}, f_{\hat{h}}) - \ell(h_*, f_{h_*})] + \mcE_n
    \end{align}
    If the empirical algorithm was minimizing $\E_n[\ell(h, f_h)]$, then the first term would have been negative and we can ignore it. However, our estimating is optimizing $\sup_{f} \E_n[2(m(W;f)-h(X)\,f(Z))-f(Z)^2]$. We now need to argue the discrepancy between $\E_n[\ell(h, f_h)]$ and $\sup_{f} \E_n[2(m(W;f)-h(X)\,f(Z))-f(Z)^2]$.

    Let $\hat{f}_{h}=\argmax_{f\in \Fcal}\E_n[2(m(W;f)-h(X)\,f(Z))-f(Z)^2]$, then our algorithm is optimizing $\E_n[\ell(h, \hat{f}_h)]$. We consider the discrepancy:
    \begin{align}
        \E_n[\ell(h, \hat{f}_h) - \ell(h, f_h)]
    \end{align}
    Let $Q_{h}(\tau)=\E\left[\ell\left(h, f_h + \tau (\hat{f}_h - f_h)\right)\right]$. Since $f_h$ is an interior first order optimal point for the optimization problem $\sup_f \E[\ell(h,f)]$, it satisfies that $Q_{h}'(0)=0$ for any $h$. Moreover, since $\ell(h,f)$ is quadratic in $f$, we have that $Q_{h}(\tau)$ is quadratic in $\tau$, with $Q_h''(\tau) = -2\|\hat{f}_h-f_h\|^2$. By an exact second order functional Taylor expansion of $Q_h(\tau)$, we have that:
    \begin{align}
        \E[\ell(h,{f}_h) - \ell(h, \hat{f}_h)] = Q_{h}(0) - Q_h(1) = - Q_{h}'(0) - \frac{1}{2} \int_0^1 Q_{h}''(\tau) d\tau = \|\hat{f}_h - f_h\|^2  \label{eqn:quadratic}
    \end{align}

    By the optimality of $\hat{f}_h$ we also have that:
    \begin{align}
        \E_n[\ell(h,f_h) - \ell(h, \hat{f}_h)]\leq 0
    \end{align}
    Thus we conclude that:
    \begin{align}\label{eqn: f-hat-bound}
        \|\hat{f}_h - f_h\|^2 \leq (\E-\E_n)[\ell(h,{f}_h) - \ell(h, \hat{f}_h)]
    \end{align}
    If $\delta_n$ upper bounds the critical radius of $\shull(\Hcal\cdot \Fcal)$, $\shull(m\circ \Fcal)$ and $\shull(\Fcal)$, then by the localized concentration inequality in Lemma~\ref{lem:concentration}, we have, w.p. $1-\zeta/t$, for all $h\in \Hcal$:
    \begin{multline}
        \left|(\E-\E_n)[\ell(h,{f}_h) - \ell(h, \hat{f}_h)]\right| \\
        \leq O\left(\delta_n \left(\sqrt{\E[m(W; f_h-\hat{f}_h)^2] + \E[h(X)^2 (f_h(Z)-\hat{f}_h(Z))^2]} + \|f_h - \hat{f}_h\|\right) + \delta_n^2\right)
    \end{multline}
    By mean-squared-continuity of $m$ and boundedness of the functions:
    \begin{align}\label{eq:strong-concentration}
        \left|(\E-\E_n)[\ell(h,{f}_h) - \ell(h, \hat{f}_h)]\right| \leq~& O\left(\delta_n \|f_h - \hat{f}_h\| + \delta_n^2\right) \leq \frac{1}{2} \|f_h - \hat{f}_h\|^2 + O\left(\delta_n^2\right)
    \end{align}
    Combining Equation~\eqref{eq:strong-concentration} with Equation~\eqref{eqn: f-hat-bound} and re-arranging yields:
        \begin{align}\label{eq:strong}
        \E[\ell(h,{f}_h) - \ell(h, \hat{f}_h)] = \|\hat{f}_h - f_h\|^2 \leq O\left(\delta_n^2\right)
    \end{align}
    Subsequently we also get for all $h\in \Hcal$:
    \begin{align}
        \E_n[\ell(h,{f}_h) - \ell(h, \hat{f}_h)] \leq \E[\ell(h,{f}_h) - \ell(h, \hat{f}_h)] + \frac{1}{2} \|f_h - \hat{f}_h\|^2 + O\left(\delta_n^2\right) = O\left(\delta_n^2\right)
    \end{align}

    Applying this fast concentration to each leading term in Equation~\eqref{eqn:middle-step}:
    \begin{align}
        \|\Tcal(h_0 - \hat{h})\|^2 - \|\Tcal(h_0 - h_*)\|^2\leq \E_n[\ell(\hat{h}, \hat{f}_{\hat{h}}) - \ell(h_*, \hat{f}_{h_*})] + \mcE_n + O(\delta_n^2)
    \end{align}
    Since by the definition of our estimation algorithm, $\hat{h}$ minimizes:
    \begin{align}
        \E_n[\ell({h}, \hat{f}_{{h}})] + R_n(h)
    \end{align}
    we get that the leading term is at most $R_n(h_*)-R_n(\hat{h})$. This concludes the proof.
    \qed\end{proof}
    
    Applying Lemma~\ref{lem:weak-conv-tikh} with $R_n(h)=\lambda\|h-\hat{h}_{-1}\|_{2,n}^2$ and combining with the strong convexity lower bound we get:
    \begin{align}
        \lambda \|\hat{h}-h_*\|^2 +  \|\Tcal(\hat{h} - h_*)\|^2\leq~& \|\Tcal(h_0 - \hat{h})\|^2 - \|\Tcal(h_0 - h_*)\|^2  + \lambda \left(\|\hat{h}\|^2 - \|h_*\|^2\right) \\
        \leq~& O\left(\delta_n \|\hat{h}-h_*\|\, \|f_{h_*}\|_{L_\infty} + \delta_n \|\Tcal(\hat{h} - h_*)\| + \delta_n^2\right) + \lambda\,(I_n - I)
    \end{align}
    Applying the AM-GM inequality to the leading terms and re-arranging, yields:
    \begin{align}
        \frac{\lambda}{2} \|\hat{h}-h_*\|^2 = O\left(\frac{\delta_n^2 \|f_{h_*}\|^2_{L_\infty}}{\lambda} + \delta_n^2\right) + \lambda\, (I_n - I)
    \end{align}

    Consider the discrepancy between the empirical and the ideal regularizers, i.e. $I_n-I$. The two differ first in that they use different centering functions, i.e. $\hat{h}_{-1}$ vs. $h_{*,-1}$ and in using empirical vs. population $L_2$ norms. We separate the two errors by writing: 
    \begin{align}
        I_n - I = I_n - \hat{I} + \hat{I} - I
    \end{align}
    where $\hat{I} = \|h_*-\hat{h}_{-1}\|^2 - \|\hat{h}-\hat{h}_{-1}\|^2$.

    \noindent\textbf{Regularization error due to centering error} We first analyze $\hat{I}-I$.
    \begin{align}
    \hat{I} - I = \|h_*-\hat{h}_{-1}\|^2 - \|\hat{h}-\hat{h}_{-1}\|^2 - \left(\|h_*-h_{*,-1}\|^2 - \|\hat{h}-h_{*,-1}\|^2\right)
    \end{align}
    The function $\|h_*-h\|^2 - \|\hat{h}-h\|^2$, simplifies to:
    \begin{align}
        \|h_*\|^2 - \|\hat{h}\|^2 + 2\langle h, \hat{h} - h_*\rangle
    \end{align}
    Thus $\hat{I}-I$ simplifies to:
    \begin{align}
        |\hat{I}-I| = \left|2 \langle \hat{h}_{-1} - h_{*,-1}, \hat{h} - h_*\rangle\right| \leq 2 \|\hat{h}_{-1} - h_{*,-1}\|\, \|\hat{h}-h_*\|
    \end{align}
    
    \noindent\textbf{Empirical vs. population regularization} Next consider:
    \begin{align}
        I_n - \hat{I} :=~& \|h_*-\hat{h}_{-1}\|_{2,n}^2 - \|\hat{h}-\hat{h}_{-1}\|_{2,n}^2 - (\|h_*-\hat{h}_{-1}\|^2 - \|\hat{h}-\hat{h}_{-1}\|^2)
    \end{align}
    We further split this into two differences of centered empirical processes:
    \begin{align}
        I_1 :=~& \|h_*-\hat{h}_{-1}\|_{2,n}^2 - \|h_*-{h}_{*,-1}\|_{2,n}^2 - (\|h_*-\hat{h}_{-1}\|^2 - \|h_*-h_{*,-1}\|^2)\\
        I_2 :=~& \|h_*-h_{*,-1}\|_{2,n}^2 - \|\hat{h}-\hat{h}_{-1}\|_{2,n}^2 - (\|h_*-h_{*,-1}\|^2 - \|\hat{h}-\hat{h}_{-1}\|^2)
    \end{align}
    Since each of these is the difference of two centered empirical processes, that are also Lipschitz losses (since $h_*,\hat{h},h_{*,-1}, \hat{h}_{-1}$ are uniformly bounded) and since $h_*$ 
    is a population quantity and not dependent on the empirical sample that is used for the $t$-th iterate,
    we can also upper bound these, w.p. $1-2\zeta/t$ by \begin{align}
    I_1 =~& O\left(\delta_n \|\hat{h}_{-1}-h_{*,-1}\| + \delta_n^2\right),\\
    I_2 =~& O\left(\delta_n \|\hat{h}-h_* + h_{*,-1} - \hat{h}_{-1}\| + \delta_n^2\right) = O\left(\delta_n\left(\|\hat{h}-h_*\| + \|\hat{h}_{-1}-h_{*,-1}\|\right) + \delta_n^2\right)
    \end{align}
    where $\delta_n$ is an upper bound on the critical radius of $\shull(\Hcal-\Hcal)$.
    
    Thus we get:
    \begin{align}
        \frac{\lambda}{2}  \|\hat{h}-h_*\|^2 
        \leq~& O\left(\frac{\delta_n^2 \|f_{h_*}\|^2_{L_\infty}}{\lambda} + \delta_n^2 +  \lambda \delta_n \|\hat{h}-h_*\| + \lambda \delta_n\|\hat{h}_{-1} - h_{*,-1}\|\right) + 2\lambda \|\hat{h}_{-1} - h_{*,-1}\|\, \|\hat{h}-h_*\|
    \end{align}
    
    Applying the AM-GM inequality to the last three terms:
    \begin{align}\label{eqn:middle-step-iter}
    \frac{\lambda}{8} \|\hat{h}-h_*\|^2 \leq~& O\left(\frac{\delta_n^2\|f_{h_*}\|_{L_\infty}^2}{\lambda} + \delta_n^2\right) + 2\lambda \|\hat{h}_{-1} - h_{*,-1}\|^2
    \end{align}
    
    The latter also trivially implies that:
    \begin{align}
        \|\hat{h}-h_*\|^2 \leq O\left(\frac{\delta_n^2\|f_{h_*}\|_{L_\infty}^2}{\lambda^2} + \frac{\delta_n^2}{\lambda}\right) + 16\, \|\hat{h}_{-1} - h_{*,-1}\|^2
    \end{align}
    Let $M_t = \max\left\{\lambda, \|\Tcal(h_{*,t} - h_0)\|_{L_{\infty}}^2\right\}$.
    Then the above shows that any $\tau\leq t$:
    \begin{align}
        \|\hat{h}_{\tau}-h_{*,\tau}\|^2 \leq O\left(\frac{\delta_n^2\, M_\tau}{\lambda^2}\right) + 16\, \|\hat{h}_{\tau-1} - h_{*,\tau-1}\|^2
    \end{align}
    Let $\gamma_t = \|\hat{h}_{t} - h_{*,t}\|^2$. Note that for $t=0$, since $\gamma_0=0$. Then, if we let $M_{\leq t} = \max_{\tau\leq t}M_\tau$, by our recursive bound, we have:
    \begin{align}
        \gamma_t \leq C\frac{\delta_n^2\, M_t}{\lambda^2} + 16 \gamma_{t-1} \leq C\frac{\delta_n^2\, M_{\leq t}}{\lambda^2} +  16\gamma_{t-1}
    \end{align}
    By a simple induction, this then yields the closed form bound:
    \begin{align}
        \gamma_t \leq 16^t C \frac{\delta_n^2\, M_{\leq t}}{\lambda^2}
    \end{align}
    
    By the Bias Lemma~\ref{lem:bias-tikhonov}, we have:
    \begin{align}
        \|h_{*,t} - h_0\|^2 \leq~& \|w_0\|\,  \lambda^{\min\{\beta, 2t\}} & 
        \|\Tcal(h_{*,t} - h_0)\|^2 \leq~& \|w_0\|\,  \lambda^{\min\{\beta+1, 2t\}}
    \end{align}
    Adding the bias term we get:
    \begin{align}
        \|\hat{h}_t-h_0\|^2 \leq O\left(\|\hat{h}_t - h_{*,t}\| + \|h_{*,t} -h_0\|\right)=  O\left(16^t\frac{\delta_n^2\, M_{\leq t}}{\lambda^2} + \|w_0\|\, \lambda^{\min\{\beta, 2t\}}\right)
    \end{align}

    \noindent\textbf{Weak metric rate.} 
    Applying Lemma~\ref{lem:weak-conv-tikh} with $R_n(h)=\lambda\|h-\hat{h}_{-1}\|_{2,n}^2$ we get:
    \begin{align}
        \lambda  \|\hat{h}-h_*\|^2 + \|\Tcal(\hat{h}_t - h_*)\|^2 \leq O\left(\|\Tcal(h_0 - h_*)\|^2 + \delta_n \|\Tcal (\hat{h}-h_*)\|^2 + \delta_n^2\right)  + I_n - I
    \end{align}
    Invoking the bound on $I_n-I$ and the AM-GM inequality, we get:
    \begin{align}
        \frac{\lambda}{2}  \|\hat{h}-h_*\|^2 + \frac{1}{2} \|\Tcal(\hat{h}_t - h_*)\|^2 \leq~& O\left(\|\Tcal(h_0 - h_*)\|^2 + \delta_n^2  + \lambda \|\hat{h}_{-1} - h_{*,-1}\|^2\right)\\
        \leq~& O\left(\|\Tcal(h_0 - h_*)\|^2 + \delta_n^2 + \min\left\{\lambda, \frac{16^t\delta_n^2 M_{\leq t-1}}{\lambda}\right\}\right)
    \end{align}
    Invoking the weak metric bias upper bound, we get the desired bound.
\qed\end{proof}

\section{Proof of Corollary~\ref{thm:adv-l3}}

Our goal is to show without any source condition, (1) $\|\Tcal(\tilde h-h_0)\|^2_2 = 2\mu_n + c \delta^2_n$ with probability at least $1-O(\zeta)$ for all $\tilde h\in \tilde{\Hcal}$
and (2) for any $h \in \Hcal$ such that  $\|\Tcal(h-h_0)\|^2_2 = \frac{2}{3}\max(\mu_n-c\delta^2_n,0)$, the function $h$ belongs to $\tilde \Hcal$ with probability at least $1-O(\zeta)$, for some universal constant $c$.
Then, as a corollary by setting $\mu_n \geq c'\delta^2_n$, we can ensure $\|\Tcal(\tilde h- h_0)\|^2_2 = O(\mu_n)$. Moreover, the guarantee for the strong metric under the $\beta$-source condition follows by simply invoking \cref{thm:adv-l2} and conditioning on the above events. Since all statements in \cref{thm:adv-l2} hold in high probability, the result can be concluded via a union bound. The only crucial condition we need to ensure from \cref{thm:adv-l2} is realizability, i.e. $h_* \in \tilde \Hcal$, since this is the only assumption which is affected by restricting to a sub-space of $\Hcal$. We will show that this holds by setting $\mu_n$ appropriately large.
First, we know that under the $\beta$-source condition:
\begin{align}
    \|\Tcal(h_*- h_0)\|^2_2  \leq \|w_0\|^2_2   \lambda^{\min(\beta+1, 2)} 
\end{align}
and the right hand is smaller $c''\,\lambda^{\min(\beta+1, 2)}$ for some universal constant $c''$. Thus by the second property of the first paragraph, we have that for $\mu_n\geq \frac{3}{2} \left(c''\lambda^{\min(\beta + 1, 2)} + c\delta_n^2\right)$, that $h_*$ satisfies the premise of the property and hence it also satisfies the conclusion that $h_*\in \tilde{\Hcal}$.

Thus it remains to show the two properties. Hereafter, we denote 
\begin{align}
   \hat h =  \argmin_{h\in \Hcal} \max_{f\in \Fcal} \E_n[2(Y-h(X))\,f(Z) - f(Z)^2]. 
\end{align}

\noindent\textbf{First statement: $\forall \tilde{h}\in \tilde{\Hcal}: \|\Tcal(\tilde h-h_0)\|^2_2 = \delta^2_n$ with probability $1-\zeta$.} We note that any functio $\tilde{h}\in \tilde{\Hcal}$ is an approximate optimizer of the empirical adversarial objective. We then simply note that the proof of Lemma~\ref{lem:weak-conv-tikh} can be easily modified to any $\mu_n$-approximate minimizer of the empirical criterion and not only for the exact empirical minimizer, at the expense of an extra additive $\mu_n$ in the upper bound (we omit the proof for conciseness, as it is a trivial extension). Applying then this slight extension of Lemma~\ref{lem:weak-conv-tikh} with $\hat{h}=\tilde{h}$ and $R_n(h)=0$ (and consequently with $h_*=h_0$), we get:
\begin{align}
    \|T(\tilde{h}-h_0)\|^2 \leq O\left(\delta_n \|T(\tilde{h}-h_0)\| + \delta_n^2\right) + \mu_n
\end{align}
Applying the AM-GM inequality to the first term on the right hand side and re-arranging we get:
\begin{align}
    \|T(\tilde{h}-h_0)\|^2 \leq O\left( \delta_n^2\right) + 2\mu_n
\end{align}
\noindent\textbf{Second statement. Any function $h \in \Hcal$ such that  $\|\Tcal(h-h_0)\|^2_2 \leq \frac{2}{3}\max(\mu_n - c\, \delta^2_n,0) $,  belongs to $\tilde \Hcal$ with probability at least $1-O(\zeta)$.}
Take any $h$ such that $h \in \Hcal$ such that  $\|\Tcal(h-h_0)\|^2_2 = O(\mu_n)$. Let $f_{h}=\E[h_0(X) - h(X)\mid Z=\cdot]$ and $\ell(h,f)$ as defined in Equation~\eqref{eqn:game-loss}. As noted in Equation~\eqref{eqn:game-loss-equivalence}, we can write for any $h\in \Hcal$ and assuming that $f_h\in \Fcal$, $\|\Tcal(h_0 - h)\|^2  =~ \E[\ell({h}, f_{{h}}) ] = \|f_{ h}\|^2$. Moreover, by Equation~\eqref{eq:strong-concentration} and Equation~\eqref{eq:strong} and Equation~\eqref{eqn:quadratic}, we have that for all $h\in \Hcal$:
\begin{align}
    \left|(\E-\E_n)[\ell(h,{f}_h) - \ell(h, \hat{f}_h)]\right| \leq~& O(\delta_n^2) &
    \E[\ell(h,{f}_h) - \ell(h, \hat{f}_h)] =~& \|f_h - \hat{f}_h\|^2 = O\left(\delta_n^2\right)
\end{align}
which together imply that:
\begin{align}
    \left|\E_n[\ell(h,{f}_h) - \ell(h, \hat{f}_h)]\right| = O\left(\delta_n^2\right) 
\end{align}
Applying the latter for any $h\in \Hcal$ with $\|f_h\|^2\leq \max(\mu_n - c_3 \delta_n^2, 0)$ and for $\hat{h}$, we get:
\begin{align}
     L_n(h) - L_n(\hat{h}) :=~& \E_n[\ell(h,\hat f_{h})] -  \E_n[\ell(\hat h,\hat f_{\hat h})]\\
     \leq~& \E_n[\ell(h, f_{h})] -  \E_n[\ell(\hat h, f_{\hat h})] + O\left(\delta^2_n\right)\\
     \leq~& \E[\ell(h, f_{h}) - \ell(\hat h, f_{\hat h})] + \left|(\E-\E_n)[\ell(h, f_{h})]\right| + \left|(\E-\E_n)[\ell(\hat h, f_{\hat h})]\right| + O\left(\delta^2_n\right)
\end{align}
By localized concentration, uniform boundedness and mean-squared-continuity of the moment, we also have that w.p. $1-\zeta$, for all $h\in \Hcal$:
\begin{align}
    \left|(\E-\E_n)[\ell(h, f_{h})]\right| =~& O\left(\delta_n \left(\sqrt{\E[m(W; f_{{h}} )^2]} + \sqrt{\E[(h(X)\, f_{{h}}(Z))^2]} + \|f_h\|\right) + \delta_n^2\right)\\
    =~& O\left(\delta_n\|f_h\| + \delta_n^2\right)
\end{align}
Thus we conclude that w.p. $1-O(\zeta)$:
\begin{align}
    L_n(h) - L_n(\hat{h})\leq~& \E[\ell(h, f_{h}) - \ell(\hat h, f_{\hat h})] + O\left(\delta_n \|f_h\| + \delta_n \|f_{\hat{h}}\| + \delta_n^2\right)\\
    =~& \|f_h\|^2 - \|f_{\hat{h}}\|^2 + O\left(\delta_n \|f_h\| + \delta_n \|f_{\hat{h}}\| + \delta_n^2\right)\\
    \leq~& \frac{3}{2}\|f_h\|^2 - \frac{1}{2} \|f_{\hat{h}}\|^2 +  O\left(\delta_n^2\right) \tag{AM-GM inequality}\\
    \leq~& \frac{3}{2}\|f_h\|^2 +  O\left(\delta_n^2\right)
    ~\leq~ \frac{3}{2}\|f_h\|^2 +  c\,\delta_n^2
\end{align}
for some universal constant $c$.
Moreover, by assumption we have that for any $\mu_n \geq c\delta_n^2$ that $\|f_h\|^2 \leq \frac{2}{3} \left(\mu_n - c\delta_n^2\right)$. Thus we conclude that w.p. $1-O(\zeta)$:
\begin{align}
    L_n(h) - L_n(\hat{h})\leq~&  \frac{3}{2} \frac{2}{3}\left( \mu_n - c\delta_n^2\right) + c \delta_n^2 \leq \mu_n
\end{align}
in which case $h\in \tilde{\Hcal}$. 

\section{Proof of Corollary~\ref{thm:adv-l3-iter}}

The proof of Corollary~\ref{thm:adv-l3-iter} follows along identical lines as that of Corollary~\ref{thm:adv-l3} applied to every iterate. With the same exact reasoning we can argue that for every iterate $t$: (1) $\|\Tcal(\tilde h-h_0)\|^2_2 = 2\mu_{n,t} + c \delta^2_n$ with probability at least $1-O(\zeta)$ for all $\tilde h\in \tilde{\Hcal}_t$.
and (2) for any $h \in \Hcal$ such that  $\|\Tcal(h-h_0)\|^2_2 = \frac{2}{3}\max(\mu_{n,t}-c\delta^2_n,0)$, the function $h$ belongs to $\tilde \Hcal_t$ with probability at least $1-O(\zeta)$, for some universal constant $c$.
The first property follows by noting that Lemma~\ref{lem:weak-conv-tikh-biasless} can be easily adapted to allow an additive approximation error to the empirical optimization, similar to Lemma~\ref{lem:weak-conv-tikh} that we used in the proof of Corollary~\ref{thm:adv-l3}. The second property follows with identical reasoning as the corresponding property in the proof of Corollary~\ref{thm:adv-l3}.
Then, as a corollary by setting $\mu_{n,t} \geq c'\max\{\delta^2_n, \lambda^{\min(\beta,2t)}\}$, we can ensure $\|\Tcal(\tilde h- h_0)\|^2_2 = O(\mu_{n,t}) $. Moreover, the guarantee for the strong metric under the $\beta$-source condition follows by simply invoking \cref{thm:adv-l2-iter}, using the smoothness assumption to upper bound $M_{\leq t}$ by $O(\lambda^{\gamma})$ (as in Remark~\ref{rem:smoothness}) and conditioning on the above events. Since all statements in \cref{thm:adv-l2-iter} hold in high probability, the result can be concluded via a union bound. The only crucial condition we need to ensure from \cref{thm:adv-l2-iter} is realizability, i.e. $h_{*,\tau} \in \tilde \Hcal_\tau$, for all $\tau\leq t$, since this is the only assumption which is affected by restricting to a sub-space of $\Hcal$.
We will show that this holds by setting $\mu_{t,n}$ appropriately large. First, we know that under the $\beta$-source condition for any $\tau\leq t$:
\begin{align}
    \|\Tcal(h_{*,\tau}- h_0)\|^2_2  \leq \|w_0\|^2_2 \lambda^{\min(\beta+1,2\tau)}.
\end{align}
and the right hand is smaller $c''\,\lambda^{\min(\beta+1,2\tau)}$ for some universal constant $c''$. Thus by the second property of the first paragraph, we have that for $\mu_{n,\tau}\geq \frac{3}{2} \left(c''\lambda^{\min(\beta+1,2\tau)} + c\delta_n^2\right)$, that $h_{*,\tau}$ satisfies the premise of the second property and hence it also satisfies the conclusion that $h_{*,\tau}\in \tilde{\Hcal}_\tau$.

\section{Proof of Lemma~\ref{lem:doubly_robust}}

This follows by the properties of the functions $h_0, q_0$:
\begin{align}
    \theta(h,q) - \theta_0 =~& \theta(h,q) - \theta(h_0,q_0)\\
    =~& \E[\tilde{m}(W;h) + m(W;q) - q(Z)\,h(X)] - \E[\tilde{m}(W;h_0) + m(W;q_0) - q_0(Z)\, h_0(X)]\\
    =~& \E[\tilde{m}(W;h-h_0) + m(W;q-q_0) - q(Z)\,h(X) +  q_0(Z)\, h_0(X)]\\
    =~& \E[q_0(Z)\,(h(X)-h_0(X)) + h_0(X)\,(q(Z)-q_0(Z)) - q(Z)\,h(X) +  q_0(Z)\, h_0(X)]\\
    =~& \E[q_0(Z)\,(h(X)-h_0(X)) + (h_0(X) - h(X))\,q(Z)]\\
    =~& \E[(q_0(Z) - q(Z))\,(h(X)-h_0(X))]
\end{align}

\section{Discussion on Source Condition for $q_0$}\label{app:riesz-smooth}

Note that the source condition for the moment problem that defines $q_0$, is essentially an assumption on the Riesz representer $a_0$. Since the non-parametric problem that defines $q_0$ is of the form $\Tcal^*(a_0 - q_0) = 0$, the source condition for the problem states that:
\begin{align}
    \sum_{i=1}^{\infty} \frac{\langle q_0, u_i\rangle^2}{\sigma_i^{2\beta}} <\infty
\end{align}
Moreover, note that since $\Tcal^* q_0 = a_0$, we have that:
\begin{align}
    a_0 = \sum_{i} \sigma_i \langle q_0, u_i\rangle v_i\implies
    \langle a_0, v_i\rangle = \sigma_i \langle q_0, u_i\rangle
\end{align}
Thus a $\beta$-source condition on $q_0$ can be equivalently expressed as a $(\beta+1)$-source condition on $a_0$:
\begin{align}
    \sum_{i=1}^{\infty} \frac{\langle a_0, v_i\rangle^2}{\sigma_i^{2(\beta+1)}} <\infty
\end{align}
For the existence of $q_0$, we already know that the above must hold for $\beta=0$. Hence, a $\beta$ source condition on $q_0$ is a stronger version of the source condition on $a_0$ that guarantees existence of a solution to the IV problem in Equation~\eqref{eqn:IV-q}.

This implicitly states that the Riesz representer $a_0$ should be primarily supported on the lower spectrum of the conditional expectation operator $\Tcal$. Another way of saying it, is that the functional that we care about is the inner product of the non-parametric regression function and a function that lies on the lower part of the spectrum of the operator, i.e. the functional projects the non-parametric function onto the lower part of the spectrum! 

\end{document}